\documentclass[a4paper, 12pt]{amsart}
\usepackage{amsmath}
\usepackage{amsfonts}
\usepackage{amssymb}

\newtheorem{theorem}{Theorem}[section]

\newtheorem{proposition}[theorem]{Proposition}
\newtheorem{corollary}[theorem]{Corollary}
\theoremstyle{definition}
\newtheorem{definition}[theorem]{Definition}
\newtheorem{remark}[theorem]{Remark}
\newcommand{\Tr}{\operatorname{Tr}}
\newcommand{\id}{\operatorname{id}}

\begin{document}
\date{2018-3-13}
\title{Balance between quantum Markov semigroups}
\author{Rocco Duvenhage and Machiel Snyman}
\address{Department of Physics\\
University of Pretoria\\
Pretoria 0002\\
South Africa}
\email{rocco.duvenhage@up.ac.za}

\begin{abstract}
The concept of balance between two state preserving quantum Markov semigroups
on von Neumann algebras is introduced and studied as an extension of
conditions appearing in the theory of quantum detailed balance. This is partly
motivated by the theory of joinings. Balance is defined in terms of certain
correlated states (couplings), with entangled states as a specific case. Basic
properties of balance are derived and the connection to correspondences in the
sense of Connes is discussed. Some applications and possible applications,
including to non-equilibrium statistical mechanics, are briefly explored.

\end{abstract}
\maketitle
\tableofcontents

\section{Introduction}

Motivated by quantum detailed balance, we define and study the notion of
balance between pairs of quantum Markov semigroups on von Neumann algebras,
where each semigroup preserves a faithful normal state. Ideas related to
quantum detailed balance continue to play an important role in studying
certain aspects of non-equilibrium statistical mechanics, in particular
non-equilibrium steady states. See for example \cite{AFQ}, \cite{AFQ2} and
\cite{AI}. A theory of balance as introduced here, is therefore potentially
applicable to non-equilibrium statistical mechanics. In this paper, however,
we just lay the foundations by developing the basics of a theory of balance.
Non-equilibrium is only touched on.

The papers on quantum detailed balance that most directly lead to the work
presented in this paper are \cite{FR10}, \cite{FR}, \cite{FU} and \cite{DS}.
Of particular relevance are ideas connected to standard quantum detailed
balance and maximally entangled bipartite states. Standard quantum detailed
balance conditions were mentioned in \cite{DF}, but discussed and developed in
\cite{FU} and \cite{FR}. Connections to maximally entangled states were
discussed in \cite{FR10}, \cite{FR} and \cite{DS}. However, a number of other
papers develop ideas related to standard quantum detailed balance and
dualities, of which \cite{BQ}, \cite{BQ2} and \cite{MS} contributed to our
line of investigation.

The theory of balance can be viewed as being parallel to the theory joinings
for W*-dynamical systems. The latter was developed in \cite{D, D2, D3}, and
studied further in \cite{BCM}, for the case where the dynamics are given by
$\ast$-automorphism groups. Some aspects of noncommutative joinings also
appeared in \cite{SaT} and \cite{KLP} related to entropy, and in \cite{Fid}
related to certain ergodic theorems. In \cite{KS} results closely related to
joinings were presented regarding a coupling method for quantum Markov chains
and mixing times.

The theory of joinings is already a powerful tool in classical ergodic theory
(see the book \cite{G} for an exposition), which is what motivated its study
in the noncommutative case. Analogously, we expect a theory of balance between
quantum Markov semigroups to be of use in the study of such semigroups.

The definition of balance is given in Section \ref{afdBalDef}, along with
relevant mathematical background, in particular regarding the definition of a
dual of certain positive maps. Couplings of states on two von Neumann algebras
are also defined there, essentially being states on compound systems reducing
to the states of the individual systems.

In Section \ref{afdE} we show how couplings lead to unital completely positive
(u.c.p.) maps from one von Neumann algebra to another. Of central importance
in this regard, is the diagonal coupling of two copies of the same state. In
certain standard special cases of states on the algebra $B(\mathfrak{H})$,
with $\mathfrak{H}$ a finite dimensional or separable Hilbert space, the
diagonal coupling is the maximally entangled bipartite state compatible with
the single system states (see Subsection \ref{VbDin}), indicating a close
connection between these u.c.p. maps and entanglement. These u.c.p. maps and
diagonal couplings play a key role in developing the theory of balance. This
is related to \cite[Section 4]{BCM}, although in the latter, certain
assumptions involving modular groups are built into the framework, while
analogous assumptions are not made in the definition of balance leading to the
theory developed in this paper.

Section \ref{afdKar} gives a characterization of balance in terms of
intertwinement with the u.c.p. maps defined in Section \ref{afdE}. The role of
KMS-duals and the special case of KMS-symmetry are also briefly discussed in
the context of symmetry of balance. Two simple applications are then given to
illustrate the use of balance. One is to characterize an ergodicity condition
in a way analogous to the theory of joinings (Proposition \ref{ErgKar}). The
other is on the convergence of states to steady states in open quantum systems
and non-equilibrium statistical mechanics (Proposition \ref{neig}).

The development of the theory of balance continues in Section \ref{afdTrans},
where balance is shown to be transitive, using the composition of couplings.
The definition and properties of such compositions are treated in some detail.
The connection to correspondences in the sense of Connes is also discussed.
The connection of correspondences to joinings was already pointed out in
\cite{BCM} and \cite[Section 5]{KLP}.

Next, in Section \ref{fb}, we discuss a quantum detailed balance condition
(namely standard quantum detailed balance with respect to a reversing
operation, from \cite{FU} and \cite{FR}) in terms of balance. Based on this,
we briefly speculate on non-equilibrium steady states in the context of balance.

We turn to a simple example to illustrate a number of the ideas from this
paper in Section \ref{Vb}.

In the final section, possible further directions of study are mentioned.

\section{The definition of balance}\label{afdBalDef}

This section gives the definition of balance, but for convenience and
completeness also collects some related known results that we need in the
formulation of this definition as well as later on in the paper. Some of the
notation used in the rest paper is also introduced.

In this paper we consider systems defined as follows:

\begin{definition}
\label{stelsel}A \emph{system} $\mathbf{A}=(A,\alpha,\mu)$ consists of a
faithful normal state $\mu$ on a (necessarily $\sigma$-finite) von Neumann
algebra $A$, and a unital completely positive (u.c.p.) map 
$\alpha:A\rightarrow A$, such that $\mu\circ\alpha=\mu$.
\end{definition}

\begin{remark}
Note that we only consider a single u.c.p. map, since throughout the paper we
can develop the theory at a single point in time. This can then be applied to
a semigroup of u.c.p. maps by applying the definitions and results to each
element of the semigroup separately (also see Remarks \ref{kontTyd},
\ref{kontTyd2}, \ref{ktydKMSd} and \ref{ktydTh-KMSd}, Proposition \ref{neig},
and Section \ref{Vb}).
\end{remark}

In the rest of the paper the symbols $\mathbf{A}$, $\mathbf{B}$ and
$\mathbf{C}$ will denote systems $(A,\alpha,\mu)$, 
$\left(B,\beta,\nu\right)$ and $\left(C,\gamma,\xi\right)$ respectively. The unit of a
von Neumann algebra will be denoted by $1$. When we want to emphasize it is
the unit of, say, $A$, the notation $1_{A}$ will be used.

Without loss of generality, in this paper we always assume that these von
Neumann algebras are in the cyclic representations associated with the given
states, i.e. the cyclic representation of $(A,\mu)$ is of the form 
$(G_{\mu},\id_{A},\Lambda_{\mu})$, where $G_{\mu}$ is the Hilbert space,
$\id_{A}$ denotes the identity map of $A$ into $B(G_{\mu})$, and
$\Lambda_{\mu}$ is the cyclic and separating vector such that 
$\mu(a)=\left\langle \Lambda_{\mu},a\Lambda_{\mu}\right\rangle$.

The dynamics $\alpha$ of a system $\mathbf{A}$ is necessarily a contraction,
since it is positive and unital (see for example \cite[Proposition
II.6.9.4]{Bl}). Furthermore, $\alpha$ is automatically normal. This is due to
the following result:

\begin{theorem}
\label{normaal}Let $M$ and $N$ be von Neumann algebras on the Hilbert spaces
$H$ and $K$ respectively, and consider states on them respectively given by
$\mu(a)=\left\langle \Omega,a\Omega\right\rangle $ and 
$\nu(b)=\left\langle\Lambda,b\Lambda\right\rangle $, with $\Omega\in H$ and 
$\Lambda\in K$ cyclic vectors, i.e. $\overline{M\Omega}=H$ and $\overline{N\Lambda}=K$. 
Assume that $\nu$ is faithful and consider a positive linear (but not necessarily unital)
$\eta:M\rightarrow N$ such that
\[
\nu(\eta(a)^{\ast}\eta(a))\leq\mu(a^{\ast}a)
\]
for all $a\in M$. Then it follows that $\eta$ is normal, i.e. $\sigma$-weakly 
continuous.
\end{theorem}

Results of this type appear to be well known, so we omit the proof. This
result applies to a system $\mathbf{A}$, since from the Stinespring dilation
theorem \cite{Sti} one obtains Kadison's inequality 
$\alpha(a)^{\ast}\alpha(a)\leq\alpha(a^{\ast}a)$ for all $a\in A$, i.e. $\alpha$ 
is a Schwarz mapping; see for example \cite[Proposition II.6.9.14]{Bl}.

A central notion in our work is the dual of a system, defined as follows:

\begin{definition}
\label{duaalDef}The \emph{dual} of the system $\mathbf{A}$, is the system
$\mathbf{A}^{\prime}=(A',\alpha',\mu')$ where $A'$ is the commutant of $A$ 
(in $B(G_{\mu})$), $\mu'$ is the state on $A'$ given by 
$\mu'(a')=\left\langle\Lambda_{\mu},a'\Lambda_{\mu}\right\rangle$ 
for all $a^{\prime}\in A'$, and $\alpha^{\prime}:A^{\prime}\rightarrow A^{\prime}$ 
is the unique map such that
\[
\left\langle \Lambda_{\mu},a\alpha'(a')\Lambda_{\mu}\right\rangle 
=\left\langle \Lambda_{\mu},\alpha(a)a'\Lambda_{\mu}\right\rangle
\]
for all $a\in A$ and all $a'\in A'$.
\end{definition}

Note that in this definition we have
\[
\mu^{\prime}=\mu\circ j_{\mu}
\]
where
\begin{equation}
j_{\mu}:=J_{\mu}(\cdot)^{\ast}J_{\mu} \label{j}
\end{equation}
with $J_{\mu}$ the modular conjugation associated to $\mu$.

The dual of a system is well-defined because of the following known result:

\begin{theorem}
\label{duaalBestaan}Let $H$ and $K$ be Hilbert spaces, $M$ a (not necessarily
unital) $\ast$-subalgebra of $B(H)$, and $N$ a (not necessarily unital)
C*-subalgebra of $B(K)$. Let $\Omega\in H$ with $\left\|  \Omega\right\|  =1$
be cyclic for $M$, i.e. $M\Omega$ is dense in $H$, and let $\Lambda\in K$ be
any unit vector. Set
\[
\mu:M\rightarrow\mathbb{C}:a\mapsto\left\langle \Omega,a\Omega\right\rangle
\]
and
\[
\nu:N\rightarrow\mathbb{C}:b\mapsto\left\langle \Lambda,b\Lambda\right\rangle
.
\]
Consider any positive linear $\eta:M\rightarrow N$, i.e. for a positive
operator $a\in M$, we have that $\eta(a)$ is a positive operator. Assume
furthermore that
\[
\nu\circ\eta=\mu.
\]
Then there exists a unique map, called the \emph{dual} of $\eta$,
\[
\eta^{\prime}:N^{\prime}\rightarrow M^{\prime}
\]
such that
\[
\left\langle \Omega,a\eta^{\prime}(b^{\prime})\Omega\right\rangle
=\left\langle \Lambda,\eta(a)b^{\prime}\Lambda\right\rangle
\]
for all $a\in M$ and $b^{\prime}\in N^{\prime}$. The map $\eta^{\prime}$ is
necessarily linear, positive and unital, i.e. $\eta^{\prime}(1)=1$, and
$\left\|  \eta^{\prime}\right\|  =1$. Furthermore the following two results
hold under two different sets of additional assumptions:

(a) If $\eta$ is $n$-positive, then $\eta^{\prime}$ is $n$-positive as well.
In particular, if $\eta$ is completely positive, then $\eta^{\prime}$ is as well.

(b) If $M$ and $N$ contain the identity operators on $H$ and $K$ respectively,
and $\eta$ is unital (i.e. $\eta(1)=1$), then it follows that
\[
\mu^{\prime}\circ\eta^{\prime}=\nu^{\prime},
\]
where $\mu^{\prime}(a^{\prime}):=\left\langle \Omega,a^{\prime}\Omega
\right\rangle $ and $\nu^{\prime}(b^{\prime}):=\left\langle \Lambda,b^{\prime
}\Lambda\right\rangle $ for all $a^{\prime}\in M^{\prime}$ and $b^{\prime}\in
N^{\prime}$. If in addition $\Lambda$ is separating for $N^{\prime}$, then
$\eta^{\prime}$ is faithful in the sense that when 
$\eta^{\prime}(b^{\prime\ast}b^{\prime})=0$, it follows that $b^{\prime}=0$.
\end{theorem}

\begin{proof}
This is proven using \cite[Lemma 1 on p. 53]{Di}. See \cite[Proposition
3.1]{AC} and \cite[Theorem 2.1]{AH}.
\end{proof}

Strictly speaking one should say that $\eta^{\prime}$ is the dual of $\eta$
with respect to $\mu$ and $\nu$, but the states will always be implicitly clear.

In particular, with $M=N=A$ and $\Omega=\Lambda=\Lambda_{\mu}$, we see from
this theorem that the dual of the system $\mathbf{A}$ is well-defined.

\begin{remark}
\label{kontTyd}If instead of the single map $\alpha$ we have a semigroup of
u.c.p. maps $(\alpha_{t})_{t\geq0}$ leaving $\mu$ invariant, then $\alpha
_{t}^{\prime}\equiv(\alpha_{t})^{\prime}$ also gives a semigroup of u.c.p.
maps leaving $\mu^{\prime}$ invariant. The continuity or measurability
properties of this dual semigroup (as function of $t$) will depend on those of
$\alpha_{t}$. Consider for example the standard assumption made for
(continuous time) quantum Markov semigroups, namely that $t\mapsto\alpha
_{t}(a)$ is $\sigma$-weakly continuous for every $a\in A$. Then it can be
shown that $t\mapsto\varphi(\alpha_{t}^{\prime}(a^{\prime}))$ is continuous
for every $a^{\prime}\in A^{\prime}$ and every normal state $\varphi$ on
$A^{\prime}$, so $t\mapsto\alpha_{t}^{\prime}(a^{\prime})$ is $\sigma$-weakly
continuous for every $a^{\prime}\in A^{\prime}$. I.e. $(\alpha_{t}^{\prime
})_{t\geq0}$ is also a quantum Markov semigroup (with the same type of
continuity property). If we were to include these assumptions in our
definition of a system, then the dual of such a system would therefore still
be a system. Our example in Section \ref{Vb} will indeed be for semigroups
indexed by $t\geq0$, with even stronger continuity properties. Also, see for
example the dynamical flows considered in \cite{AH}, where weaker assumptions
are made.
\end{remark}

It is helpful to keep the following fact about duals in mind:

\begin{corollary}
\label{2deDuaal}If in addition to the assumptions in Theorem
\ref{duaalBestaan} (prior to parts (a) and (b)), we have that $M$ and $N$ are
von Neumann algebras, $\eta$ is unital and $\Lambda$ is cyclic for $N^{\prime}$, then we have
\[
\eta''=\eta.
\]
\end{corollary}

\begin{proof}
This follows directly from the theorem itself, since $\eta'':M\rightarrow N$ 
is then the unique map such that 
$\left\langle\Lambda,b^{\prime}\eta^{\prime\prime}(a)\Lambda\right\rangle =
\left\langle\Omega,\eta^{\prime}(b^{\prime})a\Omega\right\rangle $ for all 
$a\in M$ and $b^{\prime}\in N^{\prime}$, while we know (again from the theorem) that
$\left\langle \Lambda,b^{\prime}\eta(a)\Lambda\right\rangle =
\left\langle\Omega,\eta^{\prime}(b^{\prime})a\Omega\right\rangle $ for all 
$a\in M$ and $b^{\prime}\in N^{\prime}$.
\end{proof}

We also record the following simple result:

\begin{proposition}
\label{joetaoj}If in Theorem \ref{duaalBestaan} we assume in addition that
$\mu$ and $\nu$ are faithful normal states on von Neumann algebras $M$ and $N$
(so $\Omega$ and $\Lambda$ are the corresponding cyclic and separating
vectors), then
\[
(j_{\nu}\circ\eta\circ j_{\mu})^{\prime}=j_{\mu}\circ\eta^{\prime}\circ
j_{\nu}
\]
for the map $j_{\nu}\circ\eta\circ j_{\mu}:M^{\prime}\rightarrow N^{\prime}$
obtained in terms of Eq. (\ref{j}).
\end{proposition}

\begin{proof}
It is a straightforward calculation to show that
\[
\left\langle \Omega,a^{\prime}j_{\mu}\circ\eta^{\prime}\circ j_{\nu}
(b)\Omega\right\rangle =\left\langle \Lambda,j_{\nu}\circ\eta\circ j_{\mu
}(a^{\prime})b\Lambda\right\rangle
\]
for all $a^{\prime}\in M^{\prime}$ and $b\in N$.
\end{proof}

This proposition is related to KMS-duals and KMS-symmetry which appear in
Sections \ref{afdKar} and \ref{fb} via the following definition:

\begin{definition}
\label{KMS-duaalDef}The map $\eta^{\sigma}:=j_{\mu}\circ\eta^{\prime}\circ
j_{\nu}:N\rightarrow M$ in Proposition \ref{joetaoj} will be referred to as
the \emph{KMS-dual} of the positive linear map $\eta:M\rightarrow N$.
\end{definition}

Combining Corollary \ref{2deDuaal} and Proposition \ref{joetaoj}, we see that
\begin{equation}
(\eta^{\sigma})^{\sigma}=\eta. \label{2deKMS-d}
\end{equation}
Further remarks and references on the origins of KMS-duals can be found in
Section \ref{afdKar}.

Let us now finally turn to our main concern in this paper:

\begin{definition}
\label{balans}Let $\mu$ and $\nu$ be faithful normal states on the von Neumann
algebras $A$ and $B$ respectively. A \emph{coupling} of $(A,\mu)$ and
$(B,\nu)$, is a state $\omega$ on the algebraic tensor product $A\odot
B^{\prime}$ such that
\[
\omega(a\otimes1)=\mu(a)\text{ \ and \ }\omega(1\otimes b^{\prime}
)=\nu^{\prime}(b^{\prime})
\]
for all $a\in A$ and $b\in B^{\prime}$. We also call such an $\omega$ a
coupling of $\mu$ and $\nu$. Let $\mathbf{A}$ and $\mathbf{B}$ be systems. We
say that $\mathbf{A}$ and $\mathbf{B}$ (in this order) are in \emph{balance}
with respect to a coupling $\omega$ of $\mu$ and $\nu$, expressed in symbols
as
\[
\mathbf{A}\omega\mathbf{B},
\]
if
\[
\omega(\alpha(a)\otimes b^{\prime})=\omega(a\otimes\beta^{\prime}(b^{\prime}))
\]
for all $a\in A$ and $b^{\prime}\in B^{\prime}$.
\end{definition}

Notice that this definition is in terms of the dual $\mathbf{B}^{\prime}$
rather than in terms of $\mathbf{B}$ itself. To define balance in terms of
$\omega(\alpha(a)\otimes b)=\omega(a\otimes\beta(b))$, for $a\in A$ and $b\in
B$, turns out to be a less natural convention, in particular with regards to
transitivity (see Section \ref{afdTrans}). Also, strictly speaking, saying
that $\mathbf{A}$ and $\mathbf{B}$ are in balance, implies a direction, say
from $\mathbf{A}$ to $\mathbf{B}$. These points will become more apparent in
subsequent sections. For example, symmetry of balance will be explored in
Section \ref{afdKar} in terms of KMS-symmetry of the dynamics $\alpha$ and
$\beta$.

\begin{remark}
\label{kontTyd2}For systems given by quantum Markov semigroups $(\alpha
_{t})_{t\geq0}$ and $(\beta_{t})_{t\geq0}$, instead of a single map for each
system, we note that balance is defined by requiring $\omega(\alpha
_{t}(a)\otimes b^{\prime})=\omega(a\otimes\beta_{t}^{\prime}(b^{\prime}))$ at
every $t\geq0$.
\end{remark}

\begin{remark}
\label{bind}For comparison to the theory of joinings \cite{D, D2, D3}, note
that a joining of systems $\mathbf{A}$ and $\mathbf{B}$, with $\alpha$ and
$\beta$ $\ast$-automorphisms, is a state $\omega$ on $A\odot B$ such that
$\omega(a\otimes1)=\mu(a)$, $\omega(1\otimes b)=\nu(b)$ and $\omega
\circ(\alpha\odot\beta)=\omega$. In addition \cite{BCM} also assumes that
$\omega\circ(\sigma_{t}^{\mu}\odot\sigma_{t}^{\nu})=\omega$, where $\sigma
_{t}^{\mu}$ and $\sigma_{t}^{\nu}$ are the modular groups associated to $\mu$
and $\nu$. In \cite{BCM}, however, it is formulated in terms of the opposite
algebra of $B$, which is in that sense somewhat closer to the conventions used
above for balance.
\end{remark}

\section{Couplings and u.c.p. maps}\label{afdE}

Here we define and study a map $E_{\omega}$ associated to a coupling $\omega$.
This map is of fundamental importance in the theory of balance, as will be
seen in the next two sections. We do not consider systems in this section, only
couplings. At the end of Section \ref{afdTrans} we discuss how $E_{\omega}$
appears in the theory of correspondences. Some aspects of this section and the
next are closely related to \cite[Section 4]{BCM} regarding joinings (see
Remark \ref{bind}).

Let $\omega$ be a coupling of $(A,\mu)$ and $(B,\nu)$ as in Definition
\ref{balans}. To clarify certain points later on in this and subsequent
sections, we consider multiple (but necessarily unitarily equivalent) cyclic
representations of a given von Neumann algebra and state. This requires us to
have corresponding notations. We assume without loss of generality that
$(B,\nu)$ is in its cyclic representation, denoted here by $(G_{\nu
},\operatorname{id}_{B},\Lambda_{\nu})$, which means that $(G_{\nu
},\operatorname{id}_{B^{\prime}},\Lambda_{\nu})$ is a cyclic representation of
$(B^{\prime},\nu^{\prime})$. Similarly, we assume that $(A,\mu)$ is in the
cyclic representation $(G_{\mu},\operatorname{id}_{A},\Lambda_{\mu})$.

Denoting the cyclic representation of $(A\odot B^{\prime},\omega)$ by
$(H_{\omega},\pi_{\omega},\Omega_{\omega})$, we obtain a second cyclic
representation $(H_{\mu},\pi_{\mu},\Omega_{\mu})$ of $(A,\mu)$ by setting
\begin{equation}
H_{\mu}:=\overline{\pi_{\omega}(A\otimes1)\Omega_{\omega}}\text{, }\pi_{\mu
}(a):=\pi_{\omega}(a\otimes1)|_{H_{\mu}}\text{ and }\Omega_{\mu}
:=\Omega_{\omega} \label{Hmu}
\end{equation}
for all $a\in A$, since
\[
\left\langle \Omega_{\mu},\pi_{\mu}(a)\Omega_{\mu}\right\rangle =\left\langle
\Omega_{\omega},\pi_{\omega}(a\otimes1)\Omega_{\omega}\right\rangle
=\omega(a\otimes1)=\mu(a).
\]
Similarly
\begin{equation}
H_{\nu}:=\overline{\pi_{\omega}(1\otimes B^{\prime})\Omega_{\omega}}\text{,
}\pi_{\nu^{\prime}}(b^{\prime}):=\pi_{\omega}(1\otimes b^{\prime})|_{H_{\nu}
}\text{ and }\Omega_{\nu}:=\Omega_{\omega} \label{Hnu}
\end{equation}
for all $b^{\prime}\in B^{\prime}$, gives a second cyclic representation
$(H_{\nu},\pi_{\nu^{\prime}},\Omega_{\nu})$ of $(B^{\prime},\nu^{\prime})$. In
particular $H_{\mu}$ and $H_{\nu}$ are subspaces of $H_{\omega}$.

We can define a unitary equivalence
\begin{equation}
u_{\nu}:G_{\nu}\rightarrow H_{\nu} \label{GnuHnu}
\end{equation}
from $(G_{\nu},$id$_{B^{\prime}},\Lambda_{\nu})$ to $(H_{\nu},\pi_{\nu
^{\prime}},\Omega_{\nu})$ by
\[
u_{\nu}b^{\prime}\Lambda_{\nu}:=\pi_{\nu^{\prime}}(b^{\prime})\Omega_{\nu}
\]
for all $b^{\prime}\in B^{\prime}$. Then
\begin{equation}
\pi_{\nu^{\prime}}(b^{\prime})=u_{\nu}b^{\prime}u_{\nu}^{\ast} \label{pi_nu'}
\end{equation}
for all $b^{\prime}\in B^{\prime}$. By setting
\begin{equation}
\pi_{\nu}(b):=u_{\nu}bu_{\nu}^{\ast} \label{pi_nu}
\end{equation}
for all $b\in B$, we also obtain a second cyclic representation $(H_{\nu}
,\pi_{\nu},\Omega_{\nu})$ of $(B,\nu)$, which has the property
\[
\pi_{\nu}(B)^{\prime}=\pi_{\nu^{\prime}}(B^{\prime})
\]
as is easily verified.

Let
\[
P_{\nu}\in B(H_{\omega})
\]
be the projection of $H_{\omega}$ onto $H_{\nu}$.

\begin{proposition}
\label{Eformule}In terms of the notation above, we have
\[
u_{\nu}^{\ast}\iota_{H_{\nu}}^{\ast}\pi_{\omega}(a\otimes1)\iota_{H_{\nu}
}u_{\nu}=u_{\nu}^{\ast}P_{\nu}\pi_{\omega}(a\otimes1)u_{\nu}\in B
\]
for all $a\in A$, where $\iota_{H_{\nu}}:H_{\nu}\rightarrow H_{\omega}$ is the
inclusion map, and $\iota_{H_{\nu}}^{\ast}:H_{\omega}\rightarrow H_{\nu}$ its adjoint.
\end{proposition}

\begin{proof}
Note that $P_{\nu}=\iota_{H_{\nu}}^{\ast}$, so indeed $u_{\nu}^{\ast}
\iota_{H_{\nu}}^{\ast}\pi_{\omega}(a\otimes1)\iota_{H_{\nu}}u_{\nu}=u_{\nu
}^{\ast}P_{\nu}\pi_{\omega}(a\otimes1)u_{\nu}$. We now show that this is in
$B$.

For any $b^{\prime}\in B^{\prime}$ we have $\pi_{\omega}(1\otimes b^{\prime
})H_{\nu}^{\bot}\subset H_{\nu}^{\bot}$, since $\pi_{\omega}(1\otimes
b^{\prime\ast})H_{\nu}\subset H_{\nu}$. It follows that $P_{\nu}\pi_{\omega
}(1\otimes b^{\prime})=\pi_{\omega}(1\otimes b^{\prime})P_{\nu}$. Therefore
\begin{align*}
P_{\nu}\pi_{\omega}(a\otimes1)|_{H_{\nu}}\pi_{\nu^{\prime}}(b^{\prime})  &
=P_{\nu}\pi_{\omega}(a\otimes1)\pi_{\omega}(1\otimes b^{\prime})|_{H_{\nu}}\\
&  =P_{\nu}\pi_{\omega}(1\otimes b^{\prime})\pi_{\omega}(a\otimes1)|_{H_{\nu}
}\\
&  =\pi_{\omega}(1\otimes b^{\prime})P_{\nu}\pi_{\omega}(a\otimes1)|_{H_{\nu}
}\\
&  =\pi_{\nu^{\prime}}(b^{\prime})P_{\nu}\pi_{\omega}(a\otimes1)|_{H_{\nu}}
\end{align*}
for all $a\in A$ and $b^{\prime}\in B^{\prime}$. So $P_{\nu}\pi_{\omega
}(a\otimes1)|_{H_{\nu}}\in\pi_{\nu^{\prime}}(B^{\prime})^{\prime}=\pi_{\nu
}(B)$. Hence $u_{\nu}^{\ast}P_{\nu}\pi_{\omega}(a\otimes1)u_{\nu}\in B$ by Eq.
(\ref{pi_nu}).
\end{proof}

This proposition proves part of the following result, which defines the
central object of this section, namely the map $E_{\omega}:A\rightarrow B$.

\begin{theorem}
\label{E-eienskap}In terms of the notation above we have the following
well-defined linear map
\begin{equation}
E_{\omega}:A\rightarrow B:a\mapsto u_{\nu}^{\ast}\iota_{H_{\nu}}^{\ast}
\pi_{\omega}(a\otimes1)\iota_{H_{\nu}}u_{\nu} \label{E}
\end{equation}
which is normal and completely positive. It has the following properties:
\[
E_{\omega}(1)=1
\]
\[
\left\|  E_{\omega}\right\|  =1
\]
\begin{equation}
\nu\circ E_{\omega}=\mu\label{nuE=mu}
\end{equation}
\end{theorem}

\begin{proof}
The map $a\mapsto\pi_{\omega}(a\otimes1)$ is completely positive, since it is
a $\ast$-homomorphism. Therefore $E_{\omega}$ is completely positive, as it is
the composition of the completely positive maps $a\mapsto\pi_{\omega}
(a\otimes1)$, $\iota_{H_{\nu}}^{\ast}(\cdot)\iota_{H_{\nu}}$ and $u_{\nu
}^{\ast}(\cdot)u_{\nu}$.

From Eq. (\ref{E}) we have $E_{\omega}(1)=$ $u_{\nu}^{\ast}\iota_{H_{\nu}
}^{\ast}\iota_{H_{\nu}}u_{\nu}=1$ as well as $\left\|  E_{\omega}\right\|
\leq1$, thus it follows that $\left\|  E_{\omega}\right\|  =1$. Furthermore,
\[
\nu\circ E_{\omega}(a)=\left\langle \Lambda_{\nu},E_{\omega}(a)\Lambda_{\nu
}\right\rangle =\left\langle \Omega_{\omega},\pi_{\omega}(a\otimes
1)\Omega_{\omega}\right\rangle =\omega(a\otimes1)=\mu(a)
\]
for all $a\in A$.

Lastly, Kadison's inequality, $E_{\omega}(a)^{\ast}E_{\omega}(a)\leq
E_{\omega}(a^{\ast}a)$, holds, since $E_{\omega}$ is a completely positive
contraction, so $\nu(E_{\omega}(a)^{\ast}E_{\omega}(a))\leq\nu(E_{\omega
}(a^{\ast}a))=\mu(a^{\ast}a)$, for all $a\in A$. Hence, $E_{\omega}$ is
normal, due to Theorem \ref{normaal}.
\end{proof}

\begin{remark}
\label{KorNor}The map $a\mapsto\pi_{\omega}(a\otimes1)$ itself can also be
shown to be normal (see for example the proof of \cite[Theorem 3.3]{BCM}).
\end{remark}

We proceed by discussing some further general properties of $E_{\omega}$ which
will be useful for us later.

The map $E_{\omega}$ is closely related to the \emph{diagonal coupling} of
$\nu$ with itself, which we now define: Let
\[
\varpi_{B}:B\odot B^{\prime}\rightarrow B(G_{\nu})
\]
be the unital $\ast$-homomorphism defined by extending $\varpi_{B}(b\otimes
b^{\prime})=bb^{\prime}$ via the universal property of tensor products. Here
$B(G_{\nu})$ is the von Neumann algebra of all bounded linear operators
$G_{\nu}\rightarrow G_{\nu}$. Now set
\begin{equation}
\delta_{\nu}(d)=\left\langle \Lambda_{\nu},\varpi_{B}(d)\Lambda_{\nu
}\right\rangle \label{diagKop}
\end{equation}
for all $d\in B\odot B^{\prime}$. Then $\delta_{\nu}$ is a coupling of $\nu$
with itself, which we call the \emph{diagonal coupling for} $\nu$. In terms of
this coupling we have the following characterization of $E_{\omega}$ which
will often be used:

\begin{proposition}
\label{UnE}The map $E_{\omega}$ is the unique function from $A$ to $B$ such
that%
\[
\omega(a\otimes b^{\prime})=\delta_{\nu}(E_{\omega}(a)\otimes b^{\prime})
\]
for all $a\in A$ and $b^{\prime}\in B^{\prime}$.
\end{proposition}

\begin{proof}
We simply calculate:
\begin{align*}
\delta_{\nu}(E_{\omega}(a)\otimes b^{\prime})  &  =\left\langle \Lambda_{\nu
},E_{\omega}(a)b^{\prime}\Lambda_{\nu}\right\rangle =\left\langle \Lambda
_{\nu},u_{\nu}^{\ast}P_{\nu}\pi_{\omega}(a\otimes1)u_{\nu}b^{\prime}
\Lambda_{\nu}\right\rangle \\
&  =\left\langle P_{\nu}\Omega_{\nu},\pi_{\omega}(a\otimes1)\pi_{\nu^{\prime}
}(b^{\prime})\Omega_{\nu}\right\rangle \\
&  =\left\langle \Omega_{\nu},\pi_{\omega}(a\otimes b^{\prime})\Omega_{\nu
}\right\rangle =\omega(a\otimes b^{\prime})
\end{align*}
for all $a\in A$ and $b^{\prime}\in B^{\prime}$. Secondly, suppose that for
some $b_{1},b_{2}\in B$ we have $\delta_{\nu}(b_{1}\otimes b^{\prime}
)=\delta_{\nu}(b_{2}\otimes b^{\prime})$ for all $b^{\prime}\in B^{\prime}$.
Then $\left\langle b_{1}^{\ast}\Lambda_{\nu},b^{\prime}\Lambda_{\nu
}\right\rangle =\left\langle b_{2}^{\ast}\Lambda_{\nu},b^{\prime}\Lambda_{\nu
}\right\rangle $ for all $b^{\prime}\in B^{\prime}$, so $b_{1}^{\ast}
\Lambda_{\nu}=b_{2}^{\ast}\Lambda_{\nu}$, since $B^{\prime}\Lambda_{\nu}$ is
dense in $G_{\nu}$. But $\Lambda_{\nu}$ is separating for $B$, hence
$b_{1}=b_{2}$. Therefore $E_{\omega}$ is indeed the unique function as stated.
\end{proof}

This has four simple corollaries:

\begin{corollary}
\label{unKop}If $\omega_{1}$ and $\omega_{2}$ are both couplings of $\mu$ and
$\nu$, then $\omega_{1}=\omega_{2}$ if and only if $E_{\omega_{1}}
=E_{\omega_{2}}$.
\end{corollary}

\begin{corollary}
The map $E_{\omega}$ is faithful in the sense that if $E_{\omega}(a^{\ast
}a)=0$, then $a=0$.
\end{corollary}

\begin{proof}
If $E_{\omega}(a^{\ast}a)=0$, then $\mu(a^{\ast}a)=\omega((a^{\ast}
a)\otimes1)=\delta_{\nu}(E_{\omega}(a^{\ast}a)\otimes1)=0$, but $\mu$ is
faithful, hence $a=0$.
\end{proof}

The latter also follows from Theorem \ref{duaalBestaan}(b) and $E_{\omega
}^{\prime\prime}=E_{\omega}$.

The next corollary is relevant when we consider cases of trivial balance, i.e.
balance with respect to $\mu\odot\nu^{\prime}$, and will be applied toward the
end of the next section, in relation to ergodicity:

\begin{corollary}
\label{trivkop}Let $\omega$ be a coupling of $(A,\mu)$ and $(B,\nu)$. If
$\omega=\mu\odot\nu^{\prime}$, then $E_{\omega}(a)=\mu(a)1_{B}$ for all $a\in
A$. Conversely, if $E_{\omega}(A)=\mathbb{C}1_{B}$, then $\omega=\mu\odot
\nu^{\prime}$.
\end{corollary}

\begin{proof}
If $\omega=\mu\odot\nu^{\prime}$, then $E_{\omega}(a)=\mu(a)1_{B}$ follows
from Proposition \ref{UnE}. Conversely, again using Proposition \ref{UnE}, if
$E_{\omega}(A)=\mathbb{C}1_{B}$, then $\omega(a\otimes b^{\prime})1_{B}
=\delta_{\nu}(E_{\omega}(a)\otimes b^{\prime})1_{B}=E_{\omega}(a)\delta_{\nu
}(1\otimes b^{\prime})=E_{\omega}(a)\nu^{\prime}(b^{\prime})$. In particular,
setting $b^{\prime}=1$, $E_{\omega}(a)=\mu(a)1_{B}$, so $\omega=\mu\odot
\nu^{\prime}$.
\end{proof}

\begin{corollary}
\label{diagVsId}We have $\omega=\delta_{\nu}$ if and only if $E_{\omega
}=\id_{B}$.
\end{corollary}

Next we point out that u.c.p. maps from $A$ to $B$ with specific additional
properties can be used to define couplings:

\begin{proposition}
\label{kopUitE}Let $\mu$ and $\nu$ be faithful normal states on the von
Neumann algebras $A$ and $B$ respectively. Consider a linear map
$E:A\rightarrow B$ and define a linear functional $\omega_{E}:A\odot
B^{\prime}\rightarrow\mathbb{C}$ by
\[
\omega_{E}:=\delta_{\nu}\circ(E\odot\id_{B^{\prime}}),
\]
i.e.
\[
\omega_{E}(a\otimes b^{\prime})=\delta_{\nu}(E(a)\otimes b^{\prime})
\]
for all $a\in A$ and $b\in B^{\prime}$. Then $\omega_{E}$ is a coupling of
$\mu$ and $\nu$ if and only if $E$ is completely positive, unital and
$\nu\circ E=\mu$. In this case $E=E_{\omega_{E}}$.
\end{proposition}

\begin{proof}
Consider a completely positive linear map $E:A\rightarrow B$. Then
$E\odot\operatorname{id}_{B^{\prime}}$ is positive, so $\omega_{E}$ is
positive, since $\delta_{\nu}$ is. If we furthermore assume that $E$ is
unital, then $\omega_{E}(1\otimes1)=1$, so $\omega_{E}$ is a state. Assuming
in addition that $\nu\circ E=\mu$, we conclude that $\omega_{E}(a\otimes
1)=\nu(E(a))=\mu(a)$ and $\omega_{E}(1\otimes b^{\prime})=$ $\nu^{\prime
}(b^{\prime})$, so $\omega_{E}$ is indeed a coupling of $\mu$ and $\nu$.
Because of Proposition \ref{UnE} we necessarily have $E=E_{\omega_{E}}$. The
converse is covered by Theorem \ref{E-eienskap} and Proposition \ref{UnE}.
\end{proof}

So in effect we can define couplings as maps $E$ of the form described in this proposition.

Lastly we study the dual $E_{\omega}^{\prime}$ of $E_{\omega}$, given by
Theorem \ref{duaalBestaan}. Given a coupling $\omega$ of $\mu$ and $\nu$, we
define
\[
\omega':=\delta_{\mu'}\circ(E_{\omega}'\odot\id_{A}):B'\odot A\rightarrow\mathbb{C}
\]
where
$\delta_{\mu'}(d'):=\left\langle\Lambda_{\mu},\varpi_{A'}(d')\Lambda_{\mu}\right\rangle$ 
for all
$d^{\prime}\in A^{\prime}\odot A$, i.e. 
$\delta_{\mu'}(a'\otimes a)=\left\langle \Lambda_{\mu},a'a\Lambda_{\mu}\right\rangle$. 
Since $E_{\omega}^{\prime}$ is a u.c.p. map, it then follows, using Theorem
\ref{duaalBestaan}, Proposition \ref{kopUitE} and Proposition \ref{UnE}, that
$\omega^{\prime}$ is a coupling of $\nu^{\prime}$ and $\mu^{\prime}$ such
that
\begin{equation}
\omega^{\prime}(b^{\prime}\otimes a)=\omega(a\otimes b^{\prime})
\label{kopAksent}
\end{equation}
for all $a\in A$ and $b^{\prime}\in B^{\prime}$.

\begin{proposition}
\label{Eomega'}In terms of the above notation we have
\[
E_{\omega}^{\prime}=E_{\omega^{\prime}}:B^{\prime}\rightarrow A^{\prime}
\]
and
\[
E_{\omega^{\prime}}(b^{\prime})=u_{\mu}^{\ast}\iota_{H_{\mu}}^{\ast}
\pi_{\omega}(1\otimes b^{\prime})\iota_{H_{\mu}}u_{\mu}%
\]
for all $b^{\prime}\in B^{\prime}$, where $u_{\mu}:G_{\mu}\rightarrow H_{\mu}$
is the unitary operator defined by
\[
u_{\mu}a\Lambda_{\mu}:=\pi_{\mu}(a)\Omega_{\mu}
\]
for all $a\in A$, $\iota_{H_{\mu}}:H_{\mu}\rightarrow H_{\omega}$ is the
inclusion map, and $\iota_{H_{\mu}}^{\ast}:H_{\omega}\rightarrow H_{\mu}$ its adjoint.
\end{proposition}

\begin{proof}
That $E_{\omega}^{\prime}=E_{\omega^{\prime}}$, follows from the definition of
$\omega^{\prime}$ and Proposition \ref{UnE} applied to $\omega^{\prime}$ and
$\delta_{\mu^{\prime}}$ instead of $\omega$ and $\delta_{\nu}$.

Note that $u_{\mu}$ is defined in perfect analogy to $u_{\nu}$ in Eq.
(\ref{GnuHnu}): As the cyclic representation of $(B^{\prime}\odot
A,\omega^{\prime})$ we can use $(H_{\omega},\pi_{\omega^{\prime}}
,\Omega_{\omega})$ with $\pi_{\omega^{\prime}}$ defined via
\[
\pi_{\omega^{\prime}}(b^{\prime}\otimes a):=\pi_{\omega}(a\otimes b^{\prime})
\]
(and the universal property of tensor products) for all $b^{\prime}\in
B^{\prime}$ and $a\in A$. Then, referring to the form of Eq. (\ref{Hnu}), we
see that in the place of $(H_{\nu},\pi_{\nu^{\prime}},\Omega_{\nu})$ we have
$(H_{\mu},\pi_{\mu},\Omega_{\mu})$, as we would expect, since $\overline
{\pi_{\omega^{\prime}}(1\otimes A)\Omega_{\omega}}=\overline{\pi_{\omega
}(A\otimes1)\Omega_{\omega}}=H_{\mu}$, $\pi_{\omega^{\prime}}(1\otimes
a)|_{H_{\mu}}=\pi_{\omega}(a\otimes1)|_{H_{\mu}}=\pi_{\mu}(a)$ and
$\Omega_{\mu}=\Omega_{\omega}$ for all $a\in A$.

So $u_{\mu}$ plays the same role for $E_{\omega^{\prime}}$ as $u_{\nu}$ does
for $E_{\omega}$, i.e. by definition (see Theorem \ref{E-eienskap})
\[
E_{\omega^{\prime}}(b^{\prime})=u_{\mu}^{\ast}\iota_{H_{\mu}}^{\ast}
\pi_{\omega^{\prime}}(b^{\prime}\otimes1)\iota_{H_{\mu}}u_{\mu}=u_{\mu}^{\ast
}\iota_{H_{\mu}}^{\ast}\pi_{\omega}(1\otimes b^{\prime})\iota_{H_{\mu}}u_{\mu}
\]
for all $b^{\prime}\in B^{\prime}$.
\end{proof}

We are now in a position to apply $E_{\omega}$ to balance in subsequent
sections. Also see Section \ref{afdVw} for brief remarks on how $E_{\omega}$
may be related to ideas from quantum information.

\section{A characterization of balance}\label{afdKar}

In this section we derive a characterization of balance in terms of the map
$E_{\omega}$ from the previous section, and consider some of its consequences,
including a condition for symmetry of balance in terms of KMS-symmetry. This
gives insight into the meaning and possible applications of balance. We
continue with the notation from Section \ref{afdE}.

The dynamics $\alpha$ of a system $\mathbf{A}$ can be represented by a
contraction $U$ on $H_{\mu}$ defined as the unique extension of
\begin{equation}
U\pi_{\mu}(a)\Omega_{\mu}:=\pi_{\mu}(\alpha(a))\Omega_{\mu} \label{U}
\end{equation}
for $a\in A$. Note that $U$ is indeed a contraction, since from Kadison's
inequality mentioned in Section \ref{afdBalDef}, we have $\mu(\alpha(a)^{\ast
}\alpha(a))\leq\mu(a^{\ast}a)$. (It is also simple to check from the
definition of the dual system that $U^{\ast}$ is the corresponding
representation of $\alpha^{\prime}$ on $H_{\mu}$.) Similarly
\[
V\pi_{\nu}(b)\Omega_{\nu}:=\pi_{\nu}(\beta(b))\Omega_{\nu}
\]
for all $b\in B$, to represent $\beta$ on $H_{\nu}$ by the contraction $V$.

Also set
\begin{equation}
P_{\omega}:=P_{\nu}|_{H_{\mu}}:H_{\mu}\rightarrow H_{\nu}, \label{Pomega}
\end{equation}
where $P_{\nu}$ is again the projection of $H_{\omega}$ onto $H_{\nu}$. Note
that from Eqs. (\ref{E}) and (\ref{pi_nu}) it follows that
\begin{equation}
P_{\omega}\pi_{\mu}(a)\Omega_{\mu}=\pi_{\nu}(E_{\omega}(a))\Omega_{\nu}
\label{Evoorst}
\end{equation}
for all $a\in A$, so $P_{\omega}$ is a Hilbert space representation of
$E_{\omega}$.

The characterization of balance in terms of $E_{\omega}$ is the following:

\begin{theorem}
\label{balkar}For systems $\mathbf{A}$ and $\mathbf{B}$, let $\omega$ be a
coupling of $\mu$ and $\nu$. Then $\mathbf{A}\omega\mathbf{B}$, i.e.
$\mathbf{A}$ and $\mathbf{B}$ are in balance with respect to $\omega$, if and
only if
\[
E_{\omega}\circ\alpha=\beta\circ E_{\omega}
\]
holds, or equivalently, if and only if $P_{\omega}U=VP_{\omega}$.
\end{theorem}

\begin{proof}
We prove it on Hilbert space level. Note that $P_{\omega}$ as defined in Eq.
(\ref{Pomega}) is the unique function $H_{\mu}\rightarrow H_{\nu}$ such that
$\left\langle P_{\omega}x,y\right\rangle =\left\langle x,y\right\rangle $ for
all $x\in H_{\mu}$ and $y\in H_{\nu}$. (This is a Hilbert space version of
Proposition \ref{UnE}, but it follows directly from the definition of
$P_{\omega}$.)

Assume that $\mathbf{A}$ and $\mathbf{B}$ are in balance with respect to
$\omega$. Then, for $x=\pi_{\mu}(a)\Omega_{\omega}\in H_{\mu}$ and $y=\pi
_{\nu^{\prime}}(b^{\prime})\Omega_{\omega}\in H_{\nu}$, where $a\in A$ and
$b^{\prime}\in B^{\prime}$,
\begin{align*}
\left\langle P_{\omega}Ux,y\right\rangle  &  =\left\langle Ux,y\right\rangle
=\left\langle \pi_{\omega}(\alpha(a)\otimes1)\Omega_{\omega},\pi_{\omega
}(1\otimes b^{\prime})\Omega_{\omega}\right\rangle \\
&  =\left\langle \Omega_{\omega},\pi_{\omega}(\alpha(a^{\ast})\otimes
b^{\prime})\Omega_{\omega}\right\rangle =\omega(\alpha(a^{\ast})\otimes
b^{\prime})\\
&  =\omega(a^{\ast}\otimes\beta^{\prime}(b^{\prime}))=\left\langle \pi
_{\omega}(a\otimes1)\Omega_{\omega},\pi_{\omega}(1\otimes\beta^{\prime
}(b^{\prime}))\Omega_{\omega}\right\rangle \\
&  =\left\langle x,V^{\ast}y\right\rangle =\left\langle P_{\omega}x,V^{\ast
}y\right\rangle =\left\langle VP_{\omega}x,y\right\rangle
\end{align*}
which implies that $P_{\omega}U=VP_{\omega}$. Therefore, using Eqs. (\ref{E}),
(\ref{Hmu}) and (\ref{pi_nu}), and since $u_{\nu}\Lambda_{\nu}=\Omega_{\omega
}$,
\begin{align*}
E_{\omega}\circ\alpha(a)\Lambda_{\nu}  &  =u_{\nu}^{\ast}P_{\omega}\pi_{\mu
}(\alpha(a))\Omega_{\omega}=u_{\nu}^{\ast}P_{\omega}U\pi_{\mu}(a)\Omega
_{\omega}\\
&  =u_{\nu}^{\ast}VP_{\omega}\pi_{\mu}(a)\Omega_{\omega}=u_{\nu}^{\ast}
Vu_{\nu}E_{\omega}(a)u_{\nu}^{\ast}\Omega_{\omega}\\
&  =u_{\nu}^{\ast}V\pi_{\nu}(E_{\omega}(a))\Omega_{\omega}=u_{\nu}^{\ast}
\pi_{\nu}(\beta\circ E_{\omega}(a))\Omega_{\omega}\\
&  =\beta\circ E_{\omega}(a)\Lambda_{\nu}
\end{align*}
but since $\Lambda_{\nu}$ is separating for $B$, this means that $E_{\omega
}\circ\alpha(a)=\beta\circ E_{\omega}(a)$.

Conversely, if $E_{\omega}\circ\alpha=\beta\circ E_{\omega}$, then by Eq.
(\ref{Evoorst}),
\begin{align*}
P_{\omega}U\pi_{\mu}(a)\Omega_{\mu}  &  =P_{\omega}\pi_{\mu}(\alpha
(a))\Omega_{\omega}=\pi_{\nu}(E_{\omega}(\alpha(a)))\Omega_{\omega}\\
&  =\pi_{\nu}(\beta\circ E_{\omega}(a))\Omega_{\omega}=V\pi_{\nu}(E_{\omega
}(a))\Omega_{\omega}\\
&  =VP_{\omega}\pi_{\mu}(a)\Omega_{\mu}
\end{align*}
so $P_{\omega}U=VP_{\omega}$. Therefore, similar to the beginning of this
proof,
\[
\omega(\alpha(a^{\ast})\otimes b^{\prime})=\left\langle P_{\omega
}Ux,y\right\rangle =\left\langle VP_{\omega}x,y\right\rangle =\omega(a^{\ast
}\otimes\beta^{\prime}(b^{\prime}))
\]
for all $a\in A$ and $b^{\prime}\in B^{\prime}$, as required.
\end{proof}

\begin{remark}
This theorem can be compared to the case of joinings in \cite[Theorems 4.1 and
4.3]{BCM}. Keep in mind that in \cite{BCM} the dynamics of systems are given
by $\ast$-automorphisms, and secondly an additional assumption is made
involving the modular groups (see Remark \ref{bind}). The u.c.p. map obtained
in \cite{BCM} from a joining then also intertwines the modular groups, not
just the dynamics. See \cite{BCM16} for closely related results.
\end{remark}

\begin{remark}
From Theorem \ref{balkar} one starts to see some aspects of the meaning of
balance. In particular it can be seen from $E_{\omega}\circ\alpha=\beta\circ
E_{\omega}$ that part of the dynamics of $\mathbf{B}$, more precisely the
restriction $\beta|_{E_{\omega}(A)}:E_{\omega}(A)\rightarrow E_{\omega}(A)$ to
the space $E_{\omega}(A)$, is given by the dynamics of $\mathbf{A}$, 
via $E_\omega$.

Furthermore, regarding the condition $P_{\omega}U=VP_{\omega}$,
we can point the reader to the papers \cite{NSZ, DSt, Pand}, which show how
the asymptotic properties of contractions on Hilbert spaces (one of the most
well-studied topics in operator theory) could be used to obtain mixing and
ergodic properties of the completely positive maps that these contractions
implement spatially. This hints at the importance of balance in ergodic
theory, in particular with regards to ergodic properties which are at least 
partially shared by two semigroups that are in balance.
\end{remark}

A natural question is whether or not balance is symmetric. I.e., are
$\mathbf{A}$ and $\mathbf{B}$ in balance with respect to $\omega$ if and only
if $\mathbf{B}$ and $\mathbf{A}$ are in balance with respect to some coupling
(related in some way to $\omega$)? Below we derive balance conditions
equivalent to $\mathbf{A}\omega\mathbf{B}$, but where (duals of) the systems
$\mathbf{A}$ and $\mathbf{B}$ appear in the opposite order. This is then used
to find conditions under which balance is symmetric.

As before, let
\[
j_{\mu}:B(G_{\mu})\rightarrow B(G_{\mu}):a\mapsto J_{\mu}a^{\ast}J_{\mu},
\]
where as in the previous section we assume that $(A,\mu)$ is in the cyclic
representation $(G_{\mu},\operatorname{id}_{A},\Lambda_{\mu})$ and $J_{\mu}$
is the corresponding modular conjugation. Similarly for $j_{\nu}$.

Given a coupling $\omega$ of $\mu$ and $\nu$, this allows us to define
\[
\omega^{\sigma}:=\delta_{\mu}\circ(E_{\omega}^{\sigma}\odot\id
_{A^{\prime}}):B\odot A^{\prime}\rightarrow\mathbb{C},
\]
where
\[
E_{\omega}^{\sigma}:=j_{\mu}\circ E_{\omega}^{\prime}\circ j_{\nu}:B\rightarrow A
\]
is the KMS-dual of $E_{\omega}$ as in Definition \ref{KMS-duaalDef}, and
$\delta_{\mu}(d):=\left\langle \Lambda_{\mu},\varpi_{A}(d)\Lambda_{\mu
}\right\rangle $ for all $d\in A\odot A^{\prime}$, i.e. $\delta_{\mu}(a\otimes
a^{\prime})=\left\langle \Lambda_{\mu},aa^{\prime}\Lambda_{\mu}\right\rangle
$. Since $j_{\mu}$ is a anti-$\ast$-automorphism, the conjugate linear map
$j_{\mu}^{\ast}:B(G_{\mu})\rightarrow B(G_{\mu})$ obtained by composing
$j_{\mu}$ with the involution, i.e.
\[
j_{\mu}^{\ast}(a):=j_{\mu}(a^{\ast})
\]
for all $a\in B(G_{\mu})$, is completely positive in the sense that if it is
applied entry-wise to elements of the matrix algebra $M_{n}(A)$, then it maps
positive elements to positive elements for every $n$, just like complete
positivity of linear maps. It follows that $E_{\omega}^{\sigma}=j_{\mu}^{\ast
}\circ E_{\omega}^{\prime}\circ j_{\nu}^{\ast}$ is a u.c.p. map, since
$E_{\omega}^{\prime}$ is. Consequently, since $\mu\circ E_{\omega}^{\sigma
}=\mu^{\prime}\circ E_{\omega}^{\prime}\circ j_{\nu}=\nu^{\prime}\circ j_{\nu
}=\nu$, it follows from Proposition \ref{kopUitE} that $\omega^{\sigma}$ is a
coupling of $\nu$ and $\mu$. It is then also clear that
\begin{equation}
E_{\omega^{\sigma}}=E_{\omega}^{\sigma} \label{Eomegasigma}
\end{equation}
by applying Proposition \ref{UnE}.

The KMS-dual of $\alpha$ is given by
\begin{equation}
\alpha^{\sigma}=j_{\mu}\circ\alpha^{\prime}\circ j_{\mu} \label{KMS-duaal}
\end{equation}
and similarly for $\beta$. This means that
\[
\left\langle \Lambda_{\mu},a_{1}j_{\mu}(\alpha^{\sigma}(a_{2}))\Lambda_{\mu
}\right\rangle =\left\langle \Lambda_{\mu},\alpha(a_{1})j_{\mu}(a_{2}
)\Lambda_{\mu}\right\rangle
\]
for all $a_{1},a_{2}\in A$, which corresponds to the definition of the
KMS-dual given in \cite[Section 2]{FR}, in connection with quantum detailed
balance. (In \cite{FR}, however, the KMS-dual is indicated by a prime rather
than the symbol $\sigma$.) Also see \cite{Pet} and \cite[Proposition
8.3]{OPet}. In the latter the KMS-dual is defined in terms of the modular
conjugation as well, as is done above, rather than just in terms of an
analytic continuation of the modular group, as is often done in other sources
(including \cite{FR}).

\begin{proposition}
\label{KMS-dstelsel}In terms of the notation above,
\[
\mathbf{A}^{\sigma}:=(A,\alpha^{\sigma},\mu)
\]
is a system, called the \emph{KMS-dual} of $\mathbf{A}$.
\end{proposition}

\begin{proof}
Simply note that $\alpha^{\sigma}$ is indeed a u.c.p. map (by the same
argument as for $E_{\omega}^{\sigma}$ above) such that $\mu\circ\alpha
^{\sigma}=\mu^{\prime}\circ\alpha^{\prime}\circ j_{\mu}=\mu^{\prime}\circ
j_{\mu}=\mu$.
\end{proof}

\begin{remark}
\label{ktydKMSd}For a QMS $(\alpha_{t})_{t\geq0}$ with the $\sigma$-weak
continuity property as in Remark \ref{kontTyd}, we again have that the same
$\sigma$-weak continuity property holds for $(\alpha_{t}^{\sigma})_{t\geq0}$
as well, where $\alpha_{t}^{\sigma}:=(\alpha_{t})^{\sigma}$ for every $t$.
This follows from the corresponding property of $(\alpha_{t}^{\prime})_{t\geq0}$.
\end{remark}

In terms of this notation, we have the following consequence of Theorem
\ref{balkar}:

\begin{corollary}
For systems $\mathbf{A}$ and $\mathbf{B}$, let $\omega$ be a coupling of $\mu$
and $\nu$. Then
\[
\mathbf{A}\omega\mathbf{B}
\Leftrightarrow 
\mathbf{B}^{\prime}\mathbb{\omega}^{\prime}\mathbf{A}^{\prime}
\Leftrightarrow 
\mathbf{B}^{\sigma}\mathbb{\omega}^{\sigma}\mathbf{A}^{\sigma}.
\]
\end{corollary}

\begin{proof}
By the definition of the dual\ of a map in Theorem \ref{duaalBestaan} (which
tells us that $(E_{\omega}\circ\alpha)^{\prime}=\alpha^{\prime}\circ
E_{\omega}^{\prime}$, etc.), as well as Proposition \ref{Eomega'} and Eqs.
(\ref{Eomegasigma}) and (\ref{KMS-duaal}), we have
\[
E_{\omega}\circ\alpha=\beta\circ E_{\omega}\Leftrightarrow E_{\omega^{\prime}
}\circ\beta^{\prime}=\alpha^{\prime}\circ E_{\omega^{\prime}}\Leftrightarrow
E_{\omega^{\sigma}}\circ\beta^{\sigma}=\alpha^{\sigma}\circ E_{\omega^{\sigma
}}
\]
which completes the proof by Theorem \ref{balkar}.
\end{proof}

This is not quite symmetry of balance. However, we say that the system
$\mathbf{A}$ (and also $\alpha$ itself) is \emph{KMS-symmetric} when
\begin{equation}
\alpha^{\sigma}=\alpha\label{KMS-sim}
\end{equation}
holds. If both $\alpha$ and $\beta$ are KMS-symmetric, then we see that
\[
\mathbf{A}\omega\mathbf{B\Leftrightarrow B}\mathbb{\omega}^{\sigma}
\mathbf{A},
\]
which expresses symmetry of balance in this special case.

KMS-symmetry was studied in \cite{GL93}, \cite{GL} and \cite{C}, and in
\cite{FU} it was considered in the context of the structure of generators of
norm-continuous quantum Markov semigroups on $B(\mathfrak{H})$ and standard
quantum detailed balance conditions.

We have however not excluded the possibility that there is some coupling other
than $\omega^{\sigma}$ that could be used to show symmetry of balance more
generally. This possibility seems unlikely, given how natural the foregoing
arguments and constructions are.

We end this section by studying some simple applications of balance that
follow from Theorem \ref{balkar} and the facts derived in the previous section.

First we consider \emph{ergodicity} of a system $\mathbf{B}$, which we define
to mean
\begin{equation}
B^{\beta}:=\{b\in B:\beta(b)=b\}=\mathbb{C}1_{B} \label{erg}
\end{equation}
in analogy to the case for $\ast$-automorphisms instead of u.c.p. maps. This
is certainly not the only notion of ergodicity available; see for example
\cite{AH} for an alternative definition which implies Eq. (\ref{erg}), because
of \cite[Lemma 2.1]{AH}. The definition we give here is however convenient to
illustrate how balance can be applied: this form of ergodicity can be
characterized in terms of balance, similar to how it is done in the theory of
joinings (see \cite[Theorem 3.3]{D}, \cite[Theorem 2.1]{D2} and \cite[Theorem
6.2]{BCM}), as we now explain.

\begin{definition}
A system $\mathbf{B}$ is said to be \emph{disjoint} from a system $\mathbf{A}$
if the only coupling $\omega$ with respect to which $\mathbf{A}$ and
$\mathbf{B}$ (in this order) are in balance, is the trivial coupling
$\omega=\mu\odot\nu^{\prime}$.
\end{definition}

In the next result, an \emph{identity system} is a system $\mathbf{A}$ with
$\alpha=\operatorname{id}_{A}$.

\begin{proposition}
\label{ErgKar}A system is ergodic if and only if it is disjoint from all
identity systems.
\end{proposition}

\begin{proof}
Suppose $\mathbf{B}$ is ergodic and $\mathbf{A}$ an identity system. If
$\mathbf{A}\omega\mathbf{B}$ for some coupling $\omega$, then $\beta\circ
E_{\omega}=E_{\omega}$ by Theorem \ref{balkar}. So $E_{\omega}(A)=\mathbb{C}1_{B}$, 
since $\mathbf{B}$ is ergodic. By Corollary \ref{trivkop} we conclude
that $\omega=\mu\odot\nu^{\prime}$.

Conversely, suppose that $\mathbf{B}$ is disjoint from all identity systems.
Recall that $A:=B^{\beta}$ is a von Neumann algebra (see for example
\cite[Lemma 6.4]{BCM} for a proof). Therefore $\mathbf{A}
:=(A,\operatorname{id}_{A},\mu)$ is an identity system, where $\mu:=\nu|_{A}$.
Define a coupling of $\mu$ and $\nu$ by $\omega:=\delta_{\nu}|_{A\odot
B^{\prime}}$ (see Eq. (\ref{diagKop})), then from Proposition \ref{UnE} we
have $E_{\omega}=\operatorname{id}_{A}$. So $E_{\omega}\circ\alpha
=\operatorname{id}_{A}=\beta\circ E_{\omega}$, implying that $\mathbf{A}$ and
$\mathbf{B}$ are in balance with respect to $\omega$ by Theorem \ref{balkar}.
Hence, by our supposition and Corollary \ref{trivkop}, $B^{\beta}=E_{\omega
}(A)=\mathbb{C}1_{B}$, which means that $\mathbf{B}$ is ergodic.
\end{proof}

It seems plausible that some other ergodic properties can be similarly
characterized in terms of balance, but that will not be pursued further in
this paper.

Our second application is connected to non-equilibrium statistical mechanics,
in particular the convergence of states to steady states. See for example the
early papers \cite{Sp}, \cite{Fri} and \cite{M83} on the topic, as well as
more recent papers like \cite{M01}, \cite{Fel} and \cite{FR08}. To clarify the
connection between these results (which are expressed in terms of continuous
time $t\geq0$) and the result below, we formulate the latter in terms of
continuous time as well. Compare it in particular to results in \cite[Section
3]{Fri}. It is an example of how properties of one system can be partially
carried over to other systems via balance.

\begin{proposition}
\label{neig}Assume that $\mathbf{A}$ and $\mathbf{B}$ are in balance with
respect to $\omega$. Suppose that
\[
\lim_{t\rightarrow\infty}\varkappa(\alpha_{t}(a))=\mu(a)
\]
for all normal states $\varkappa$ on $A$, and all $a\in A$. Then
\[
\lim_{t\rightarrow\infty}\lambda(\beta_{t}(b))=\nu(b)
\]
for all normal states $\lambda$ on $B$, and all $b\in E_{\omega}(A)$.
\end{proposition}

\begin{proof}
Applying Theorem \ref{balkar} and setting $\varkappa:=\lambda\circ E_{\omega}
$, we have
\[
\lim_{t\rightarrow\infty}\lambda(\beta_{t}(E_{\omega}(a)))=\lim_{t\rightarrow
\infty}\varkappa(\alpha_{t}(a))=\mu(a)=\nu(E_{\omega}(a))
\]
for all $a\in A$, by Theorem \ref{E-eienskap}.
\end{proof}

We expect various results of this sort to be possible, namely where two
systems are in balance, and properties of the one then necessarily hold in a
weaker form for the other.

Conversely, one can in principle use balance as a way to impose less stringent
alternative versions of a given property, by requiring a system to be in
balance with another system having the property in question. We expect that
such conditions need not be directly comparable (and strictly weaker) than the
property in question. This idea will be discussed further in relation to
detailed balance in Section \ref{fb}.

\section{Composition of couplings and transitivity of balance}\label{afdTrans}

Here we show transitivity of balance: if $\mathbf{A}$ and $\mathbf{B}$ are in
balance with respect to $\omega$, and $\mathbf{B}$ and $\mathbf{C}$ are in
balance with respect to $\psi$, then $\mathbf{A}$ and $\mathbf{C}$ are in
balance with respect to a certain coupling obtained from $\omega$ and $\psi$,
and denoted by $\omega\circ\psi$. The coupling $\omega\circ\psi$ is the
composition of $\omega$ and $\psi$, as defined and discussed in detail below.
Furthermore, we discuss the connection between couplings and correspondences
in the sense of Connes.

Let $\omega$ be a coupling of $(A,\mu)$ and $(B,\nu)$, and let $\psi$ be a
coupling of $(B,\nu)$ and $(C,\xi)$. Note that $E_{\psi}\circ E_{\omega
}:A\rightarrow C$ is a u.c.p. map such that $\xi\circ E_{\psi}\circ E_{\omega
}=\mu$ by Theorem \ref{E-eienskap}. Therefore, by Proposition \ref{kopUitE},
setting
\begin{equation}
\omega\circ\psi:=\delta_{\xi}\circ((E_{\psi}\circ E_{\omega})\odot\id_{C'}),
\label{kopKring}
\end{equation}
i.e.
\[
\omega\circ\psi(a\otimes c^{\prime})=\delta_{\xi}(E_{\psi}(E_{\omega
}(a))\otimes c^{\prime})
\]
for all $a\in A$ and $c\in C^{\prime}$, we obtain a coupling $\omega\circ\psi$
of $\mu$ and $\xi$ such that
\begin{equation}
E_{\omega\circ\psi}=E_{\psi}\circ E_{\omega}. \label{EkringE}
\end{equation}
This construction forms the foundation for the rest of this section.

We call the coupling $\omega\circ\psi$ the \emph{composition} of the couplings
$\omega$ and $\psi$. We can view it as an analogue of a construction appearing
in the theory of joinings in classical ergodic theory; see for example
\cite[Definition 6.9]{G}.

We can immediately give the main result of this section, namely that we have
transitivity of balance in the following sense:

\begin{theorem}
\label{trans}If $\mathbf{A}\omega\mathbf{B}$ and $\mathbf{B}\psi\mathbf{C}$,
then $\mathbf{A}(\omega\circ\psi)\mathbf{C}$.
\end{theorem}

\begin{proof}
By Theorem \ref{balkar} we have $E_{\omega}\circ\alpha=\beta\circ E_{\omega}$
and $E_{\psi}\circ\beta=\gamma\circ E_{\psi}$, so
\[
E_{\omega\circ\psi}\circ\alpha=E_{\psi}\circ\beta\circ E_{\omega}=\gamma\circ
E_{\omega\circ\psi},
\]
which again by Theorem \ref{balkar} means that $\mathbf{A}(\omega\circ
\psi)\mathbf{C}$.
\end{proof}

In order to gain a deeper understanding of the transitivity of balance, we now
study properties of the composition of couplings.

\begin{proposition}
\label{idVirKring}The diagonal coupling $\delta_{\nu}$ in Eq. (\ref{diagKop})
is the identity for composition of couplings in the sense that $\delta_{\nu
}\circ\psi=\psi$ and $\omega\circ\delta_{\nu}=\omega$.
\end{proposition}

\begin{proof}
By Corollary \ref{diagVsId}, $E_{\delta_{\nu}}=\id_{B}$. Hence,
from Eq. (\ref{EkringE}), we obtain 
$E_{\delta_{\nu}\circ\psi}=E_{\psi}\circ E_{\delta_{\nu}}=E_{\psi}$ and 
$E_{\omega\circ\delta_{\nu}}=E_{\delta_{\nu}}\circ E_{\omega}=E_{\omega}$, 
which concludes the proof by Corollary
\ref{unKop}.
\end{proof}

In order to treat further properties of $\omega\circ\psi$ and the connection
with the theory of correspondences, we need to set up the relevant notation:

Continuing with the notation in the previous two sections, also assuming
$(C,\xi)$ to be in its cyclic representation 
$(G_{\xi},\operatorname{id}_{C},\Lambda_{\xi})$, 
and denoting the cyclic representation of $(B\odot C^{\prime},\psi)$ by 
$(K_{\psi},\varphi_{\psi},\Psi_{\psi})$, it follows that
\[
K_{\nu}:=\overline{\pi_{\psi}(B\otimes1)\Psi_{\psi}}\text{, }\varphi_{\nu
}(b):=\varphi_{\psi}(b\otimes1)|_{K_{\nu}}\text{ and }\Psi_{\nu}:=\Psi_{\psi}
\]
gives a third cyclic representation $(K_{\nu},\varphi_{\nu},\Lambda_{\nu})$ of
$(B,\nu)$, and that
\begin{equation}
K_{\xi}:=\overline{\pi_{\psi}(1\otimes C^{\prime})\Psi_{\psi}}\text{, }
\varphi_{\xi^{\prime}}(c^{\prime}):=\varphi_{\psi}(1\otimes c^{\prime
})|_{K_{\xi}}\text{ and }\Psi_{\xi}:=\Psi_{\psi} \label{Kxi}
\end{equation}
gives a cyclic representation $(K_{\xi},\varphi_{\xi^{\prime}},\Psi_{\xi})$ of
$(C^{\prime},\xi^{\prime})$. Note that to help keep track of where we are, we
use the symbol $K$ instead of $H$ for the Hilbert spaces originating from
$\psi$ (as opposed to $\omega$), and similarly we use $\varphi$ instead of
$\pi$, and $\Psi$ instead of $\Omega$.

We can define a unitary equivalence
\begin{equation}
v_{\nu}:G_{\nu}\rightarrow K_{\nu} \label{GnuKnu}
\end{equation}
from $(G_{\nu},\operatorname{id}_{B},\Lambda_{\nu})$ to $(K_{\nu},\varphi
_{\nu},\Psi_{\nu})$ by
\[
v_{\nu}b\Lambda_{\nu}:=\varphi_{\nu}(b)\Psi_{\nu}
\]
for all $b\in B$. Then
\[
\varphi_{\nu}(b):=v_{\nu}bv_{\nu}^{\ast}
\]
for all $b\in B$.

By Theorem \ref{E-eienskap} we can then define the normal u.c.p. map
$E_{\psi^{\prime}}:C^{\prime}\rightarrow B^{\prime}$. By Proposition
\ref{Eomega'} this map is the dual $E_{\psi}^{\prime}$ of $E_{\psi}$, and we
can write it as
\begin{equation}
E_{\psi}^{\prime}:C^{\prime}\rightarrow B^{\prime}:c^{\prime}\mapsto v_{\nu
}^{\ast}\iota_{K_{\nu}}^{\ast}\varphi_{\psi}(1\otimes c^{\prime})\iota
_{K_{\nu}}v_{\nu}=v_{\nu}^{\ast}Q_{\nu}\varphi_{\psi}(1\otimes c^{\prime
})v_{\nu} \label{E'}
\end{equation}
where $Q_{\nu}$ is the projection of $K_{\psi}$ onto $K_{\nu}$, and 
$Q_{\nu}=\iota_{K_{\nu}}^{\ast}$ with $\iota_{K_{\nu}}:K_{\nu}\rightarrow K_{\psi}$
the inclusion map, in analogy to $P_{\nu}=\iota_{H_{\nu}}^{\ast}$ in
Proposition \ref{Eformule}.

The coupling $\omega\circ\psi$ can now be expressed in various ways:

\begin{proposition}
\label{KringItvE}The coupling $\omega\circ\psi$ is given by the following
formulas:
\begin{equation}
\omega\circ\psi=\delta_{\nu}\circ(E_{\omega}\odot E_{\psi}^{\prime})
\label{KringDef}
\end{equation}
and
\[
\omega\circ\psi=\delta_{\mu}\circ(\id_{A}\odot(E_{\omega
}^{\prime}\circ E_{\psi}^{\prime}))
\]
in terms of Eq. (\ref{diagKop}), as well as
\begin{equation}
\omega\circ\psi(a\otimes c^{\prime})=\psi(E_{\omega}(a)\otimes c^{\prime
})=\omega(a\otimes E_{\psi}^{\prime}(c^{\prime})) \label{KringItvPsiOmega}
\end{equation}
and
\begin{equation}
\omega\circ\psi(a\otimes c^{\prime})=\left\langle u_{\nu}^{\ast}P_{\nu}
\pi_{\mu}(a^{\ast})\Omega_{\omega},v_{\nu}^{\ast}Q_{\nu}\varphi_{\xi^{\prime}
}(c^{\prime})\Psi_{\psi}\right\rangle \label{KringItvInpr}
\end{equation}
(in the inner product of the Hilbert space $G_{\nu}$) for all $a\in A$ and
$c^{\prime}\in C^{\prime}$.
\end{proposition}

\begin{proof}
From Eqs. (\ref{kopKring}) and (\ref{diagKop}), and Theorem \ref{duaalBestaan}, 
we have
\begin{align}
\omega\circ\psi(a\otimes c^{\prime})  &  =\left\langle \Lambda_{\xi},E_{\psi
}(E_{\omega}(a))c^{\prime}\Lambda_{\xi}\right\rangle \nonumber\\
&  =\left\langle \Lambda_{\nu},E_{\omega}(a)E_{\psi}^{\prime}(c^{\prime
})\Lambda_{\nu}\right\rangle \label{kringItvEE'}
\end{align}
from which Eq. (\ref{KringDef}) follows. Continuing with the last expression
above, we respectively have by Theorem \ref{duaalBestaan} that
\begin{align*}
\omega\circ\psi(a\otimes c^{\prime})  &  =\left\langle \Lambda_{\mu
},aE_{\omega}^{\prime}(E_{\psi}^{\prime}(c^{\prime}))\Lambda_{\mu
}\right\rangle \\
&  =\delta_{\mu}\circ(\operatorname{id}_{A}\odot(E_{\omega}^{\prime}\circ
E_{\psi}^{\prime}))(a\otimes c^{\prime}),
\end{align*}
by Proposition \ref{UnE} that
\[
\omega\circ\psi(a\otimes c^{\prime})=\omega(a\otimes E_{\psi}^{\prime
}(c^{\prime}))
\]
and by Proposition \ref{Eomega'} that
\begin{align*}
\omega\circ\psi(a\otimes c^{\prime})  
&  =\left\langle \Lambda_{\nu},E_{\psi^{\prime}}(c^{\prime})E_{\omega}(a)\Lambda_{\nu}\right\rangle \\
&  =\psi^{\prime}(c^{\prime}\otimes E_{\omega}(a))\\
&  =\psi(E_{\omega}(a)\otimes c^{\prime}),
\end{align*}
where in the second line we again applied Proposition \ref{UnE}, while the
last line follows from the definition of $\psi^{\prime}$, as in Eq.
(\ref{kopAksent}).

On Hilbert space level we again have from Eq. (\ref{kringItvEE'}) that
\begin{align*}
\omega\circ\psi(a\otimes c^{\prime})  &  =\left\langle E_{\omega}(a^{\ast
})\Lambda_{\nu},E_{\psi}^{\prime}(c^{\prime})\Lambda_{\nu}\right\rangle \\
&  =\left\langle u_{\nu}^{\ast}P_{\nu}\pi_{\omega}(a^{\ast}\otimes1)u_{\nu
}\Lambda_{\nu},v_{\nu}^{\ast}Q_{\nu}\varphi_{\psi}(1\otimes c^{\prime})v_{\nu
}\Lambda_{\nu}\right\rangle \\
&  =\left\langle u_{\nu}^{\ast}P_{\nu}\pi_{\mu}(a^{\ast})\Omega_{\omega
},v_{\nu}^{\ast}Q_{\nu}\varphi_{\xi^{\prime}}(c^{\prime})\Psi_{\psi
}\right\rangle
\end{align*}
for all $a\in A$ and $c^{\prime}\in C^{\prime}$, using Theorem
\ref{E-eienskap} (and Proposition \ref{Eformule}) as well as Eqs. (\ref{E'}),
(\ref{Hmu}) and (\ref{Kxi}).
\end{proof}

At the end of this section $\omega\circ\psi$ will also be expressed in terms
of the theory of relative tensor products of bimodules; see Corollary
\ref{kringItvRelTens}.

Next we consider triviality of transitivity, namely when $\omega\circ\psi
=\mu\odot\xi^{\prime}$, in which case we also say that the couplings $\omega$
and $\psi$ are \emph{orthogonal}, in analogy to the case of classical joinings
\cite[Definition 6.9]{G}. We first note the following:

\begin{proposition}
\label{trivKring}If either $\omega=\mu\odot\nu^{\prime}$ or $\psi=\nu\odot
\xi^{\prime}$, then $\omega\circ\psi=\mu\odot\xi^{\prime}$.
\end{proposition}

\begin{proof}
By Proposition \ref{UnE}, $E_{\mu\odot\nu^{\prime}}=\mu(\cdot)1_{B}$ and
$E_{\nu\odot\xi^{\prime}}=\nu(\cdot)1_{C}$, so $(\mu\odot\nu^{\prime}
)\circ\psi(a\otimes c^{\prime})=$ $\delta_{\xi}(\mu(a)1_{C}\otimes c^{\prime
})=$ $\mu(a)\xi^{\prime}(c^{\prime})$ and $\omega\circ(\nu\odot\xi^{\prime
})(a\otimes c^{\prime})=\delta_{\xi}(\nu(E_{\omega}(a))1_{C}\otimes c^{\prime
})=\mu(a)\xi^{\prime}(c^{\prime})$ according to Eq. (\ref{kopKring}) and
Theorem \ref{E-eienskap}.
\end{proof}

However, as will be seen by example in Subsection \ref{VbBal}, in general it
is possible that $\omega\circ\psi=\mu\odot\xi^{\prime}$ even when $\omega
\neq\mu\odot\nu^{\prime}$ and $\psi\neq\nu\odot\xi^{\prime}$. In order for
$\omega\circ\psi\neq\mu\odot\xi^{\prime}$ to hold, there has to be sufficient
``overlap'' between $\omega$ and $\psi$. The following makes this precise on
Hilbert space level and also explains the use of the term ``orthogonal'' above:

\begin{proposition}
\label{trivtranskar}We have $\omega\circ\psi=\mu\odot\xi^{\prime}$ if and only
if
\[
u_{\nu}^{\ast}[P_{\nu}H_{\mu}\ominus\mathbb{C}\Omega_{\omega}]
\perp 
v_{\nu}^{\ast}[Q_{\nu}K_{\xi}\ominus\mathbb{C}\Psi_{\psi}]
\]
in the Hilbert space $G_{\nu}$ (see Section \ref{afdE}), where $P_{\nu}$ and
$Q_{\nu}$ are the projections of $H_{\omega}$ onto $H_{\nu}$ and $K_{\psi}$
onto $K_{\nu}$ respectively, and $u_{\nu}$ and $v_{\nu}$ are the unitaries
defined above (see Eqs. (\ref{GnuHnu}) and (\ref{GnuKnu})).
\end{proposition}

\begin{proof}
In terms of the projections $P_{\Omega_{\omega}}$ and $Q_{\Psi_{\psi}}$ of
$H_{\omega}$ and $K_{\psi}$ onto $\mathbb{C}\Omega_{\omega}$ and
$\mathbb{C}\Psi_{\psi}$ respectively, we have%
\begin{align*}
&  \left\langle u_{\nu}^{\ast}P_{\Omega_{\omega}}\pi_{\mu}(a^{\ast}
)\Omega_{\omega},v_{\nu}^{\ast}Q_{\Psi_{\psi}}\varphi_{\xi^{\prime}}
(c^{\prime})\Psi_{\psi}\right\rangle \\
&  =\left\langle \left\langle \Omega_{\omega},\pi_{\mu}(a^{\ast}
)\Omega_{\omega}\right\rangle u_{\nu}^{\ast}\Omega_{\omega},\left\langle
\Psi_{\psi},\varphi_{\xi^{\prime}}(c^{\prime})\Psi_{\psi}\right\rangle v_{\nu
}^{\ast}\Psi_{\psi}\right\rangle \\
&  =\mu(a)\xi^{\prime}(c^{\prime})\left\langle \Lambda_{\nu},\Lambda_{\nu
}\right\rangle \\
&  =\mu\odot\xi^{\prime}(a\otimes c^{\prime})
\end{align*}
for all $a\in A$ and $c^{\prime}\in C^{\prime}$. In terms of $P:=P_{\nu
}-P_{\Omega_{\omega}}$ and $Q:=Q_{\nu}-Q_{\Psi_{\psi}}$, it then follows from
Eq. (\ref{KringItvInpr}) that%
\begin{align*}
&  \omega\circ\psi(a\otimes c^{\prime})-\mu\odot\xi^{\prime}(a\otimes
c^{\prime})\\
&  =\left\langle u_{\nu}^{\ast}P\pi_{\mu}(a^{\ast})\Omega_{\omega},v_{\nu
}^{\ast}Q\varphi_{\xi^{\prime}}(c^{\prime})\Psi_{\psi}\right\rangle \\
&  +\left\langle u_{\nu}^{\ast}P\pi_{\mu}(a^{\ast})\Omega_{\omega},v_{\nu
}^{\ast}Q_{\Psi_{\psi}}\varphi_{\xi^{\prime}}(c^{\prime})\Psi_{\psi
}\right\rangle +\left\langle u_{\nu}^{\ast}P_{\Omega_{\omega}}\pi_{\mu
}(a^{\ast})\Omega_{\omega},v_{\nu}^{\ast}Q\varphi_{\xi^{\prime}}(c^{\prime
})\Psi_{\psi}\right\rangle \\
&  =\left\langle u_{\nu}^{\ast}P\pi_{\mu}(a^{\ast})\Omega_{\omega},v_{\nu
}^{\ast}Q\varphi_{\xi^{\prime}}(c^{\prime})\Psi_{\psi}\right\rangle .
\end{align*}
For the last line we used $u_{\nu}^{\ast}PH_{\omega}=G_{\nu}\ominus
\mathbb{C}\Lambda_{\nu}$ and $v_{\nu}^{\ast}Q_{\Psi_{\psi}}K_{\psi}
=\mathbb{C}\Lambda_{\nu}$ to obtain the one term as zero, while the other term
is zero, since $v_{\nu}^{\ast}QK_{\psi}=G_{\nu}\ominus\mathbb{C}\Lambda_{\nu}$
and $u_{\nu}^{\ast}P_{\Omega_{\omega}}H_{\omega}=\mathbb{C}\Lambda_{\nu}$.
Therefore $\omega\circ\psi(a\otimes c^{\prime})-\mu\odot\xi^{\prime}(a\otimes
c^{\prime})$ is zero for all $a\in A$ and $c^{\prime}\in C^{\prime}$ if and
only if $u_{\nu}^{\ast}[P_{\nu}H_{\mu}\ominus\mathbb{C}\Omega_{\omega}]\perp
v_{\nu}^{\ast}[Q_{\nu}K_{\xi}\ominus\mathbb{C}\Psi_{\psi}]$.
\end{proof}

To conclude this section, we discuss bimodules and correspondences, the main
goal being to show how $\omega\circ\psi$ can be expressed in terms of the
relative tensor product of bimodules obtained from $\omega$ and $\psi$. Along
the way we get an indication of the connection between couplings and
correspondences. Also see \cite{BCM} for a related discussion of
correspondences in the context of joinings.

The theory of correspondences was originally developed by Connes, but never
published in full, although it is discussed briefly in his book \cite[Appendix
V.B]{Con}. In short, a \emph{correspondence} from one von Neumann algebra,
$M$, to another, $N$, is an $M$-$N$-bimodule (where the direction \emph{from}
$M$ \emph{to} $N$, is the convention used in this paper).

For details on the relative tensor product, see for example \cite[Section
IX.3]{T2} and \cite{Fal}, but also \cite{Sa} for some of the early work on
this topic. We only outline the most pertinent aspects of relative tensor
products, and the reader is referred to these sources, in particular
\cite[Section IX.3]{T2}, for a more systematic exposition.

As before, let
\[
j_{\nu}(b):=J_{\nu}b^{\ast}J_{\nu}
\]
for all $b\in B(G_{\nu})$, with $J_{\nu}:G_{\nu}\rightarrow G_{\nu}$ the
modular conjugation associated with $(B,\Lambda_{\nu})$. Similarly, with
$(C,\xi)$ in its cyclic representation $(G_{\xi},$id$_{C},\Lambda_{\xi})$,
let
\[
j_{\xi}(c):=J_{\xi}c^{\ast}J_{\xi}
\]
for all $c\in B(G_{\xi})$, with $J_{\xi}:G_{\xi}\rightarrow G_{\xi}$ the
modular conjugation associated with $(C,\Lambda_{\xi})$.

Given a coupling $\omega$ of $(A,\mu)$ and $(B,\nu)$ as at the beginning 
of this section, we can view $H=H_{\omega}$ as an $A$-$B$-bimodule by setting
\[
\pi_{H}(a):=\pi_{\omega}(a\otimes1)
\]
and
\[
\pi_{H}^{\prime}(b):=\pi_{\omega}(1\otimes j_{\nu}(b)),
\]
and writing
\[
axb:=\pi_{H}(a)\pi_{H}^{\prime}(b)x
\]
for all $a\in A,$ $b\in B$, and $x\in H$. As already mentioned in Remark
\ref{KorNor}, $\pi_{H}$ is normal, as required for it to give a left
$A$-module, and similarly $\pi_{H}^{\prime}$ gives a normal right action of
$B$ on $H$; again see \cite[Theorem 3.3]{BCM}. When viewing $H$ as the $A$
-$B$-bimodule thus defined, we also denote it by $_{A}H_{B}$. This module is
therefore an example of a correspondence from $A$ to $B$.

With $\psi$ a coupling of $(B,\nu)$ and $(C,\xi)$ as at the beginning of this section, 
and $(K_{\psi},\varphi_{\psi},\Psi_{\psi})$ the corresponding cyclic representation 
as before, but now using the notation $K=K_{\psi}$, we analogously 
obtain the $B$-$C$-bimodule $_{B}K_{C}$ via $\pi_{K}$ and $\pi_{K}^{\prime}$ 
given by
\[
\pi_{K}(b):=\varphi_{\psi}(b\otimes1)
\]
and
\[
\pi_{K}^{\prime}(c):=\varphi_{\psi}(1\otimes j_{\xi}(c))
\]
which enables us to write
\[
byc:=\pi_{K}(b)\pi_{K}^{\prime}(c)y
\]
for all $b\in B$, $c\in C$, and $y\in K$.

Now we form the relative tensor product (see \cite[Definition IX.3.16]{T2})
\[
_{A}X_{C}:=H\otimes_{\nu}K
\]
with respect to the faithful normal state $\nu$. This is also a Hilbert space
(its inner product will be discussed below) and, as the notation on the left
suggests, the relative tensor product is itself a $A$-$C$-bimodule. This is a
special case of \cite[Corollary IX.3.18]{T2}. The reason it works is that
since $H$ is a $A$-$B$-bimodule, any element of $\pi_{H}(A)$ can be viewed as
an element of $\mathcal{L}(H_{B})$, the space of all bounded (in the usual
sense of linear operators on Hilbert spaces) right $B$-module maps. Similarly
for the right action of $C$. So $_{A}X_{C}$ is a correspondence from $A$ to
$C$, which can be viewed as the composition of the correspondences $_{A}H_{B}$
and $_{B}K_{C}$.

As one may expect, the actions of $A$ and $C$ on $H\otimes_{\nu}K$ are given
by
\[
a(x\otimes_{\nu}y)c=(ax)\otimes_{\nu}(yc)
\]
for all $a\in A$ and $c\in C$. However, in general this does not hold for all
$x\in H$ and $y\in K$. In fact the elementary tensor $x\otimes_{\nu}y$ does
not exist for all $x\in H$ and $y\in K$. However, it does work if we restrict
either $x$ or $y$ to a certain dense subspace, say $x\in\mathfrak{D}(H,\nu)\subset H$ 
and $y\in K$. (See below for further details on the space
$\mathfrak{D}(H,\nu)$.) We correspondingly use $x\in H$ and $y\in
\mathfrak{D}^{\prime}(K,\nu)\subset K$ if we rather want to restrict $y$ to a
dense subspace of $K$.

In particular we have $\Omega_{\omega}\in\mathfrak{D}(H,\nu)$ and $\Psi_{\psi
}\in\mathfrak{D}^{\prime}(K,\nu)$, so we set
\[
\Omega:=\Omega_{\omega}\otimes_{\nu}\Psi_{\psi}\in H\otimes_{\nu}K,
\]
which we use to define a state, denoted by $\omega\diamond\psi$, on $A\odot
C^{\prime}$ as follows:
\begin{equation}
\omega\diamond\psi(d):=\left\langle \Omega,\pi_{X}(d)\Omega\right\rangle
\label{SaamstelDmvRelTens}
\end{equation}
for all $d\in A\odot C^{\prime}$, where $\pi_{X}$ is the representation of
$A\odot C^{\prime}$ on $_{A}X_{C}$ given in terms of its bimodule structure by
\[
\pi_{X}(a\otimes c^{\prime})x:=axj_{\xi}(c^{\prime})
\]
for all $x\in$ $_{A}X_{C}$. Below we show that $\omega\diamond\psi=\omega
\circ\psi$, so we have the composition of couplings expressed in terms of the
relative tensor product of bimodules, i.e. in terms of the composition of correspondences.

We first review the inner product of the relative tensor product in more
detail, in order to clarify its use below. Write
\begin{equation}
\eta_{\nu}^{\prime}(b):=j_{\nu}(b)\Lambda_{\nu}=J_{\nu}b^{\ast}\Lambda_{\nu}
\label{eta'}
\end{equation}
for all $b\in B$.

For every $x\in\mathfrak{D}(H,\nu)$, define the bounded linear operator
$L_{\nu}(x):G_{\nu}\rightarrow H$ by setting
\[
L_{\nu}(x)\eta_{\nu}^{\prime}(b)=xb\equiv\pi_{H}^{\prime}(b)x
\]
for all $b\in B$, and uniquely extending to $G_{\nu}$. We note that the space
$\mathfrak{D}(H,\nu)$ is defined to ensure that $L_{\nu}(x)$ is indeed
bounded:
\[
\mathfrak{D}(H,\nu)=
\{x\in H:\left\|  xb\right\|  \leq k_{x}\left\| \eta_{\nu}^{\prime}(b)\right\|  
\text{ for all }b\in B\text{, for some }k_{x}\geq0\}
\]
It then follows that $L_{\nu}(x_{1})^{\ast}L_{\nu}(x_{2})\in B$ for all
$x_{1},x_{2}\in\mathfrak{D}(H,\nu)$. The space $H\otimes_{\nu}K$ and its inner
product is obtained from a quotient construction such that we have
\begin{equation}
\left\langle x_{1}\otimes_{\nu}y_{1},x_{2}\otimes_{\nu}y_{2}\right\rangle
=\left\langle y_{1},\pi_{K}(L_{\nu}(x_{1})^{\ast}L_{\nu}(x_{2}))y_{2}
\right\rangle _{K} \label{defrelinprod}
\end{equation}
for $x_{1},x_{2}\in\mathfrak{D}(H,\nu)$ and $y_{1},y_{2}\in K$, where for
emphasis we have denoted the inner product of $K$ by $\left\langle \cdot
,\cdot\right\rangle _{K}$. This is the ``left'' version, but there is also a
corresponding ``right'' version of this formula for the inner product (see
\cite[Section IX.3]{T2}). It can be shown from the definition of
$\mathfrak{D}(H,\nu)$, that $\pi_{H}(a)\pi_{\nu}(b)\Omega_{\omega}\in
D(H,\nu)$ for all $a\in A$ and $b\in B$, from which in turn it follows that
$\mathfrak{D}(H,\nu)$ is dense in $H$, and that $\Omega_{\omega}
\in\mathfrak{D}(H,\nu)$. Similarly $D^{\prime}(K,\nu)$, which is defined
analogously, is dense in $K$.

From this short review of the inner product, we can show that it has the
following property:

\begin{proposition}
In $H\otimes_{\nu}K$,
\begin{equation}
\left\langle a_{1}\Omega c_{1},a_{2}\Omega c_{2}\right\rangle =\psi(E_{\omega
}(a_{1}^{\ast}a_{2})\otimes j_{\xi}(c_{2}c_{1}^{\ast})) \label{relinprod}
\end{equation}
for $a_{1},a_{2}\in A$ and $c_{1},c_{2}\in C$.
\end{proposition}

\begin{proof}
Firstly, we obtain a formula for $L_{\nu}(x)$ for elements of the form
$x=\pi_{H}(a)\pi_{\nu}(b)\Omega_{\omega}\in D(H,\nu)$, where $a\in A$ and $b$.
For all $b_{1}\in B$ we have
\begin{align*}
L_{\nu}(x)\eta_{\nu}^{\prime}(b_{1})  &  =\pi_{H}^{\prime}(b_{1})\pi_{H}
(a)\pi_{\nu}(b)\Omega_{\omega}\\
&  =\pi_{H}(a)\pi_{\nu}(b)\pi_{\nu^{\prime}}(j_{\nu}(b_{1}))\Omega_{\omega}\\
&  =\pi_{H}(a)\pi_{\nu}(b)u_{\nu}\eta_{\nu}^{\prime}(b_{1}),
\end{align*}
by Eqs. (\ref{pi_nu'}) and (\ref{eta'}), which means that
\begin{equation}
L_{\nu}(\pi_{H}(a)\pi_{\nu}(b)\Omega_{\omega})=\pi_{H}(a)\pi_{\nu}(b)u_{\nu}.
\label{L-formule}
\end{equation}
Applying the special case $L_{\nu}(\pi_{H}(a)\Omega_{\omega})=\pi_{H}
(a)u_{\nu}$ of this formula, for $a_{1},a_{2}\in A$ we have
\begin{align*}
L_{\nu}(\pi_{H}(a_{1})\Omega_{\omega})^{\ast}L_{\nu}(\pi_{H}(a_{2}
)\Omega_{\omega})  &  =u_{\nu}^{\ast}P_{\nu}\pi_{H}(a_{1}^{\ast}a_{2})u_{\nu
}\\
&  =E_{\omega}(a_{1}^{\ast}a_{2}).
\end{align*}
by Theorem \ref{E-eienskap} and Proposition \ref{Eformule}. From Eq.
(\ref{defrelinprod}) we therefore have%
\begin{align*}
\left\langle a_{1}\Omega c_{1},a_{2}\Omega c_{2}\right\rangle  &
=\left\langle \pi_{K}^{\prime}(c_{1})\Psi_{\psi},\pi_{K}(E_{\omega}
(a_{1}^{\ast}a_{2}))\pi_{K}^{\prime}(c_{2})\Psi_{\psi}\right\rangle _{K}\\
&  =\left\langle \Psi_{\psi},\pi_{K}(E_{\omega}(a_{1}^{\ast}a_{2}))\pi
_{K}^{\prime}(c_{2}c_{1}^{\ast})\Psi_{\psi}\right\rangle _{K}\\
&  =\left\langle \Psi_{\psi},\varphi_{\psi}(E_{\omega}(a_{1}^{\ast}
a_{2})\otimes j_{\xi}(c_{2}c_{1}^{\ast}))\Psi_{\psi}\right\rangle _{K}\\
&  =\psi(E_{\omega}(a_{1}^{\ast}a_{2})\otimes j_{\xi}(c_{2}c_{1}^{\ast})).
\end{align*}
\end{proof}

Now we can confirm that Eq. (\ref{SaamstelDmvRelTens}) is indeed equivalent to
the original definition Eq. (\ref{kopKring}):

\begin{corollary}
\label{kringItvRelTens}We have
\[
\omega\diamond\psi=\omega\circ\psi
\]
in terms of the definitions Eq. (\ref{SaamstelDmvRelTens}) and Eq.
(\ref{kopKring}).
\end{corollary}

\begin{proof}
From Eq. (\ref{SaamstelDmvRelTens})
\begin{align*}
\omega\diamond\psi(a\otimes c^{\prime})  &  =\left\langle \Omega,\pi
_{X}(a\otimes c^{\prime})\Omega\right\rangle =\left\langle \Omega,a\Omega
j_{\xi}(c^{\prime})\right\rangle \\
&  =\psi(E_{\omega}(a)\otimes c^{\prime}))
\end{align*}
by Eq. (\ref{relinprod}), for all $a\in A$ and $c^{\prime}\in C^{\prime}$. By
Eq. (\ref{KringItvPsiOmega}), $\omega\diamond\psi=\omega\circ\psi$.
\end{proof}

So we have $\omega\circ\psi$ expressed in terms of the vector $\Omega\in
H\otimes_{\nu}K$. Note, however, that in general $H\otimes_{\nu}K$ is not the
GNS Hilbert space for the state $\omega\circ\psi$, although the former
contains the latter. Consider for example the simple case where $\omega
=\mu\odot\nu^{\prime}$ and $\psi=\nu\odot\xi^{\prime}$. Then, by Proposition
\ref{trivKring}, $\omega\circ\psi=\mu\odot\xi^{\prime}$, and the GNS Hilbert
space obtained from this state is $G_{\mu}\otimes G_{\xi}$, whereas
$H\otimes_{\nu}K=G_{\mu}\otimes G_{\nu}\otimes G_{\xi}$.

When $(A,\mu)=(B,\nu)$ and $\omega$ is the diagonal coupling $\delta_{\nu}$ in
Eq. (\ref{diagKop}), then by \cite[Proposition IX.3.19]{T2}, $_{A}X_{C}$ is
isomorphic to $_{B}K_{C}$, so in this case the correspondence $_{A}H_{B}$ acts
as an identity from the left. Similarly from the right when $\psi$ is the
diagonal coupling. This is the correspondence version of Proposition
\ref{idVirKring}.

Lastly, by Eq. (\ref{L-formule}) we have $L_{\nu}(\Omega_{\omega}
)=\iota_{H_{\nu}}u_{\nu}$, therefore $L_{\nu}(\Omega_{\omega})^{\ast}=u_{\nu
}^{\ast}P_{\nu}$, which by Theorem \ref{E-eienskap} means that
\[
E_{\omega}(a)=L_{\nu}(\Omega_{\omega})^{\ast}\pi_{H}(a)L_{\nu}(\Omega_{\omega
})
\]
for all $a\in A$. This is the form in which $E_{\omega}$ has appeared in the
theory of correspondences, as a special case of maps of the form $a\mapsto
L_{\nu}(x)^{\ast}\pi_{H}(a)L_{\nu}(x)$ for arbitrary $x\in\mathfrak{D}(H,\nu
)$; see for example \cite[Section 1.2]{P}.

\section{Balance, detailed balance and non-equilibrium}\label{fb}

Our main goal in this section is to suggest how balance can be used to define
conditions that generalize detailed balance. We then speculate on how this may
be of value in studying non-equilibrium steady states. In order to motivate
these generalized conditions, we present a specific instance of how detailed
balance can be expressed in terms of balance. We focus on only one form of
detailed balance, namely standard quantum detailed balance with respect to a
reversing operation, as defined in \cite[Definition 3 and Lemma 1]{FU} and
\cite[Definition 1]{FR}. This form of detailed balance has only appeared in
the literature relatively recently. The origins of quantum detailed balance,
on the other hand, can be found in the papers \cite{Ag}, \cite{Al},
\cite{CWA}, \cite{KFGV} and \cite{M}.

The basic idea of this section should also apply to properties other than
detailed balance conditions, as will be explained.

We begin by noting the following simple fact in terms of the diagonal coupling
$\delta_{\mu}$ (see Eq. (\ref{diagKop})):

\begin{proposition}
\label{id-kor}A system $\mathbf{A}$ is in balance with itself with respect to
the diagonal coupling $\delta_{\mu}$, i.e. $\delta_{\mu}(\alpha(a)\otimes
a^{\prime})=\delta_{\mu}(a\otimes\alpha^{\prime}(a^{\prime}))$ for all $a\in
A$ and $a^{\prime}\in A^{\prime}$. Conversely, if two systems $\mathbf{A}$ and
$\mathbf{B}$, with $(A,\mu)=(B,\nu)$, are in balance with respect to the
diagonal coupling $\delta_{\mu}$, then $\mathbf{A}=\mathbf{B}$, i.e.
$\alpha=\beta$.
\end{proposition}

\begin{proof}
The first part is simply the definition of the dual (see Definition
\ref{duaalDef} and Theorem \ref{duaalBestaan}). The second part follows from
the uniqueness of the dual, given by Theorem \ref{duaalBestaan}; alternatively
use Theorem \ref{balkar} and Corollary \ref{diagVsId}.
\end{proof}

So, if $\mathbf{A}$ and $\mathbf{B}$ are in balance with respect to the
diagonal coupling and one of the systems has some property, then the other
system has it as well, since the systems are necessarily the same.

One avenue of investigation is therefore to define generalized versions of a
given property by demanding only that a system is in balance with another
system with the given property, with respect to a coupling (or set of
couplings) other than the diagonal coupling. In particular we then do not need
to assume that the two systems have the same algebra and state.

We demonstrate this idea below for a specific property, namely standard
quantum detailed balance with respect to a reversing operation. In order to do
so, we discuss this form of detailed balance along with $\Theta$-KMS-duals:

\begin{definition}
\label{omkeer}Consider a system $\mathbf{A}$. A \emph{reversing operation} for
$\mathbf{A}$ (or for $(A,\mu)$), is a $\ast$-antihomorphism $\Theta
:A\rightarrow A$ (i.e. $\Theta$ is linear, $\Theta(a^{\ast})=\Theta(a)^{\ast}
$, and $\Theta(a_{1}a_{2})=\Theta(a_{2})\Theta(a_{1})$) such that $\Theta
^{2}=\operatorname{id}_{A}$ and $\mu\circ\Theta=\mu$. Furthermore we define
the\emph{ }$\Theta$\emph{-KMS-dual}
\[
\alpha^{\Theta}:=\Theta\circ\alpha^{\sigma}\circ\Theta
\]
of $\alpha$ in terms of the KMS-dual $\alpha^{\sigma}=j_{\mu}\circ
\alpha^{\prime}\circ j_{\mu}$ in Eq. (\ref{KMS-duaal}).
\end{definition}

The $\Theta$-KMS-dual was introduced in \cite{BQ2} in the context of systems
on $B(\mathfrak{H})$, with $\mathfrak{H}$ a separable Hilbert space. There may
be a scarcity of examples of reversing operations for general von Neumann
algebras, but a standard example for $B(\mathfrak{H})$ is mentioned in
Subsection \ref{VbOmk}.

Using the $\Theta$-KMS-dual, we can define the above mentioned form of
detailed balance:

\begin{definition}
\label{sqdb}A system $\mathbf{A}$ satisfies \emph{standard quantum detailed
balance with respect to the reversing operation} $\Theta$ for $(A,\mu)$, or
$\Theta$\emph{-sqdb},when $\alpha^{\Theta}=\alpha$.
\end{definition}

To complete the picture, we state some straightforward properties related to
reversing operations $\Theta$ and the $\Theta$-KMS-dual:

\begin{proposition}
\label{omkEiensk}Given a reversing operation $\Theta$ for $\mathbf{A}$ as in
Definition \ref{omkeer}, we define an anti-unitary operator $\theta:G_{\mu
}\rightarrow G_{\mu}$ by extending
\[
\theta a\Lambda_{\mu}:=\Theta(a^{\ast})\Lambda_{\mu}
\]
which in particular gives $\theta^{2}=1$ and $\theta\Lambda_{\mu}=\Lambda
_{\mu}$. Then
\[
\Theta(a)=\theta a^{\ast}\theta
\]
for all $a\in A$, and consequently $\Theta$ is normal. This allows us to
define
\[
\Theta^{\prime}:A^{\prime}\rightarrow A^{\prime}:a^{\prime}\mapsto\theta
a^{\prime\ast}\theta
\]
which is the dual of $\Theta$ in the sense that
\[
\left\langle \Lambda_{\mu},a\Theta^{\prime}(a^{\prime})\Lambda_{\mu
}\right\rangle =\left\langle \Lambda_{\mu},\Theta(a)a^{\prime}\Lambda_{\mu
}\right\rangle
\]
for all $a\in A$ and $a^{\prime}\in A^{\prime}$. We also have
\[
\theta J_{\mu}=J_{\mu}\theta
\]
from which
\[
\alpha^{\Theta}=(\Theta\circ\alpha\circ\Theta)^{\sigma}
\]
and
\[
(\alpha^{\Theta})^{\Theta}=\alpha
\]
follow.
\end{proposition}

\begin{proof}
The first sentence is simple. From the definition of $\theta$ and the
properties of $\Theta$, $\theta\Lambda_{\mu}=\Lambda_{\mu}$ it follows that
\[
\theta a^{\ast}\theta b\Lambda_{\mu}=\Theta((a^{\ast}\Theta(b^{\ast}))^{\ast
})\Lambda_{\mu}=\Theta(a)b\Lambda_{\mu}
\]
for all $a,b\in A$, so $\Theta(a)=\theta a^{\ast}\theta$. Normality (i.e.
$\sigma$-weak continuity) follows from this and the definition of the $\sigma
$-weak topology. For $a\in A$ and $a^{\prime}\in A^{\prime}$ we now have
$a\theta a^{\prime}\theta=\theta\Theta(a^{\ast})a^{\prime}\theta=\theta
a^{\prime}\Theta(a^{\ast})\theta=\theta a^{\prime}\theta a$, hence $\theta
a^{\prime}\theta\in A^{\prime}$. So $\Theta^{\prime}$ is well-defined, and
that it is the dual of $\Theta$ follows easily.

Denoting the closure of the operator
\[
A\Lambda_{\mu}\rightarrow A\Lambda_{\mu}:a\Lambda_{\mu}\mapsto a^{\ast}
\Lambda_{\mu}
\]
by $S_{\mu}=J_{\mu}\Delta_{\mu}^{1/2}$, as usual in Tomita-Takesaki theory, we
obtain $S_{\mu}=\theta S_{\mu}\theta=\theta J_{\mu}\theta\theta\Delta_{\mu
}^{1/2}\theta$, hence $\theta J_{\mu}\theta=J_{\mu}$ by the uniqueness of
polar decomposition, proving $\theta J_{\mu}=J_{\mu}\theta$.

Then by definition
\begin{align*}
\alpha^{\Theta}  &  =\Theta\circ j_{\mu}\circ\alpha^{\prime}\circ j_{\mu}
\circ\Theta=j_{\mu}\circ\Theta^{\prime}\circ\alpha^{\prime}\circ\Theta
^{\prime}\circ j_{\mu}=j_{\mu}\circ(\Theta\circ\alpha\circ\Theta)^{\prime
}\circ j_{\mu}\\
&  =(\Theta\circ\alpha\circ\Theta)^{\sigma}
\end{align*}
follows. So $(\alpha^{\Theta})^{\Theta}=\Theta\circ\Theta\circ\alpha
\circ\Theta\circ\Theta=\alpha$ by Eq. (\ref{2deKMS-d}).
\end{proof}

Returning now to the main goal of this section, it will be convenient for us
to express the $\Theta$-KMS dual as a system:

\begin{proposition}
\label{thKMS-dstelsel}For a reversing operation $\Theta$ as in Definition
\ref{omkeer},
\[
\mathbf{A}^{\Theta}:=(A,\alpha^{\Theta},\mu)
\]
is a system, called the $\Theta$\emph{-KMS-dual} of $\mathbf{A}$.
\end{proposition}

\begin{proof}
Recall from Proposition \ref{KMS-dstelsel} that $\mathbf{A}^{\sigma}$ is a
system. Since $\alpha^{\sigma}$ is u.c.p., it can be checked as in Proposition
\ref{KMS-dstelsel} from $\alpha^{\Theta}=\Theta^{\ast}\circ\alpha^{\sigma
}\circ\Theta^{\ast}$, where $\Theta^{\ast}(a):=\Theta(a^{\ast})$ for all $a\in
A$, that $\alpha^{\Theta}$ is u.c.p. as well. From $\mu\circ\Theta=\mu$, we
obtain $\mu\circ\alpha^{\Theta}=\mu$.
\end{proof}

\begin{remark}
\label{ktydTh-KMSd}Similar to before, for a QMS $(\alpha_{t})_{t\geq0}$ with
the $\sigma$-weak continuity property as in Remark \ref{kontTyd}, we have that
this continuity property also holds for $(\alpha_{t}^{\Theta})_{t\geq0}$,
where $\alpha_{t}^{\Theta}:=(\alpha_{t})^{\Theta}$ for every $t$. This follows
from the continuity of $(\alpha_{t}^{\sigma})_{t\geq0}$ in Remark
\ref{ktydKMSd}, and the fact that $\Theta$ is normal (Proposition
\ref{omkEiensk}).
\end{remark}

As a simple corollary of Proposition \ref{id-kor} we have:

\begin{corollary}
\label{fb&balans}Let $\mathbf{A}$ be a system and let $\Theta$ be a reversing 
operation for $\mathbf{A}$. Then the following are equivalent:

(a) $\mathbf{A}$ satisfies $\Theta$-sqdb.

(b) $\mathbf{A}$ and $\mathbf{A}^{\Theta}$ are in balance with respect to
$\delta_{\mu}$.

(c) $\mathbf{A}^{\Theta}$ and $\mathbf{A}$ are in balance with respect to
$\delta_{\mu}$.
\end{corollary}

When two systems are in balance, we expect the one system to partially inherit
properties of the other. We saw an example of this in Proposition \ref{neig}.
As mentioned there, this suggests that for any given property that a system
may have, we can in principle consider generalized forms of the property via
balance. In particular for $\Theta$-sqdb:

\begin{itemize}
\item We can consider systems $\mathbf{A}$ and $\mathbf{B}$ which are in
balance with respect to a coupling $\omega$ (or a set of couplings) other than
$\mu\odot\nu^{\prime}$, but not necessarily with respect to $\delta_{\mu}$.
Assuming that either $\mathbf{A}$ or $\mathbf{B}$ satisfies $\Theta$-sqdb, for
some reversing operation $\Theta$ for $\mathbf{A}$ or $\mathbf{B}$
respectively, the other system can then be viewed as satisfying a generalized
version of $\Theta$-sqdb.
\end{itemize}

A second possible way of obtaining conditions generalizing $\Theta$-sqdb for a
system $\mathbf{A}$, is simply to adapt Corollary \ref{fb&balans} more directly:

\begin{itemize}
\item We can require $\mathbf{A}$ and $\mathbf{A}^{\Theta}$ to be in balance
with respect to some coupling $\omega$ (or a set of couplings) other than
$\mu\odot\mu^{\prime}$, but not necessarily with respect to $\delta_{\mu}$. Or
$\mathbf{A}^{\Theta}$ and $\mathbf{A}$ to be in balance with respect to some
coupling $\omega$ (or a set of couplings) other than $\mu\odot\mu^{\prime}$,
but not necessarily with respect to $\delta_{\mu}$.
\end{itemize}

Under KMS-symmetry (see Eq. (\ref{KMS-sim})), the two options in the second
condition, namely $\mathbf{A}$ and $\mathbf{A}^{\Theta}$ in balance, versus
$\mathbf{A}^{\Theta}$ and $\mathbf{A}$ in balance, are equivalent:

\begin{proposition}
If the system $\mathbf{A}$ is KMS-symmetric, then $\mathbf{A}\omega
\mathbf{A}^{\Theta}$ if and only if $\mathbf{A}^{\Theta}\omega_{E}\mathbf{A}$,
where $E:=\Theta\circ E_{\omega}\circ\Theta$. (See Proposition \ref{kopUitE}
for $\omega_{E}$.)
\end{proposition}

\begin{proof}
By KMS-symmetry $\alpha^{\Theta}=\Theta\circ\alpha\circ\Theta$. Note that for
any coupling $\omega$ we have that $E=\Theta^{\ast}\circ E_{\omega}\circ
\Theta^{\ast}$ is u.c.p. like $\alpha^{\Theta}$ in the proof of Proposition
\ref{thKMS-dstelsel}, and $\mu\circ E=\mu$ by Theorem \ref{E-eienskap} and
$\mu\circ\Theta=\mu$. Then $\omega_{E}$ is a coupling by Proposition
\ref{kopUitE}. From Theorem \ref{balkar} we have
\[
\mathbf{A}\omega\mathbf{A}^{\Theta}\Leftrightarrow E_{\omega}\circ
\alpha=\Theta\circ\alpha\circ\Theta\circ E_{\omega}\Leftrightarrow
E\circ\alpha^{\Theta}=\alpha\circ E\Leftrightarrow\mathbf{A}^{\Theta}
\omega_{E}\mathbf{A.}
\]
\end{proof}

The two types of conditions suggested above will be illustrated by a simple
example in the next section, where the conditions obtained will in fact be
weaker than $\Theta$-sqdb.

A basic question we now have is the following: can generalized conditions like
these be applied to characterize certain non-equilibrium steady states $\mu$
which have enough structure that one can successfully analyse them
mathematically, while also having physical relevance? This seems plausible,
given that these conditions are structurally so closely related to detailed
balance itself. We briefly return to this in Section \ref{afdVw}.

\section{An example}\label{Vb}

In this section we use a very simple example based on the examples in
\cite[Section 6]{AFQ}, \cite{BQ}, \cite[Section 5]{FR10} and \cite[Subsection
7.1]{FR} to illustrate some of the ideas discussed in this paper. Our main
reason for considering this example is that it is comparatively easy to
manipulate mathematically. We leave a more in depth study of relevant examples
for future work.

Let $\mathfrak{H}$ be a separable Hilbert space with total orthonormal set
$e_{1},e_{2},e_{3},...$. We are going to consider systems on the von Neumann
algebra $B(\mathfrak{H})$. These systems will all have the same faithful
normal state $\zeta$ on $B(\mathfrak{H})$ given by the diagonal (in the
mentioned basis) density matrix
\[
\rho=\left[
\begin{array}
[c]{ccc}
\rho_{1} &  & \\
& \rho_{2} & \\
&  & \ddots
\end{array}
\right]
\]
where $\rho_{1},\rho_{2},\rho_{3},...>0$ satisfy $\sum_{n=1}^{\infty}\rho
_{n}=1$. I.e.
\[
\zeta(a)=\operatorname{Tr}(\rho a)
\]
for all $a\in B(\mathfrak{H})$.

We now briefly explain what the cyclic representation and modular conjugation
look like for the state $\zeta$:

The (faithful) cyclic representation of $(B(\mathfrak{H}),\zeta)$ can be
written as $(H,\pi,\Omega)$ where $H=\mathfrak{H}\otimes\mathfrak{H}$,
\[
\pi(a)=a\otimes1
\]
for all $a\in B(\mathfrak{H})$, and the maximally entangled state (reducing to
$\rho$)
\[
\Omega=\sum_{n=1}^{\infty}\sqrt{\rho_{n}}e_{n}\otimes e_{n}
\]
is the cyclic vector. Our von Neumann algebra is therefore represented as
\[
A=\pi(B(\mathfrak{H})),
\]
and the state $\zeta$ is represented by the state $\mu$ on $A$ given by
\[
\mu(\pi(a))=\zeta(a)
\]
for all $a\in A$. However, we also consider a second representation
$\pi^{\prime}$ given by
\[
\pi^{\prime}(a)=1\otimes a
\]
for all $a\in B(\mathfrak{H})$, so $A^{\prime}=\pi^{\prime}(B(\mathfrak{H}))$.
The state $\mu^{\prime}$ on $A^{\prime}$ is then given by
\[
\mu^{\prime}(\pi^{\prime}(a))=\left\langle \Omega,\pi^{\prime}(a)\Omega
\right\rangle =\zeta(a)
\]
for all $a\in A$.

The modular conjugation $J$ associated to $\mu$ (and to $\zeta$) is then
obtained as the conjugate linear operator $J:H\rightarrow H$ given by
\[
J(e_{p}\otimes e_{q})=e_{q}\otimes e_{p}
\]
for all $p,q=1,2,3,...$. Furthermore,
\begin{equation}
j(\pi(a)):=J\pi(a)^{\ast}J=\pi^{\prime}(a^{T}) \label{jB(h)}
\end{equation}
for all $a\in B(\mathfrak{H})$, where $a^{T}$ denotes the transpose of $a$ in
the basis $e_{1},e_{2},e_{3},...$.

This allows us to apply the general notions from the earlier sections
explicitly to this specific case.

Regarding notation: Instead of the notation $\left|  x\right\rangle
\left\langle y\right|  $ for $x,y\in\mathfrak{H}$, we use $x\Join y$, i.e.
\[
(x\Join y)z:=x\left\langle y,z\right\rangle
\]
for all $z\in\mathfrak{H}$.

\subsection{The couplings}\label{VbKop}

We consider couplings of $\zeta$ with itself. A coupling of $\zeta$ with
itself corresponds to a coupling of $\mu$ with itself in the cyclic
representation, which is a state $\omega$ on $A\odot A^{\prime}=\pi
(B(\mathfrak{H}))\odot\pi^{\prime}(B(\mathfrak{H}))\cong B(\mathfrak{H})\odot
B(\mathfrak{H})$ such that
\[
\omega(\pi(a)\otimes1)=\mu(\pi(a))\text{ \ \ and \ \ }\omega(1\otimes
\pi^{\prime}(a))=\mu^{\prime}(\pi^{\prime}(a))
\]
for all $a\in B(\mathfrak{H})$. However, in this concrete example it is
clearly equivalent, and simpler in terms of notation, to view $\omega$
directly as a state on $B(\mathfrak{H})\odot B(\mathfrak{H})$ such that
\begin{equation}
\omega(a\otimes1)=\zeta(a)\text{ \ \ and \ \ }\omega(1\otimes a)=\zeta(a)
\label{vbKop1}
\end{equation}
for all $a\in B(\mathfrak{H})$, rather than to work via the cyclic representation.

Consider any disjoint subsets $Y_{1},Y_{2},Y_{3},...$ of 
$\mathbb{N}_{+}:=\{1,2,3,4,...\}$ such that $\cup_{n=1}^{\infty}Y_{n}=\mathbb{N}_{+}$. 
We construct a coupling $\omega$ which is given by a density matrix $\kappa\in
B(\mathfrak{H}\otimes\mathfrak{H})$, i.e.
\[
\omega(c)=\Tr(\kappa c)
\]
for all $c\in B(\mathfrak{H})\odot B(\mathfrak{H})$. Therefore we may as well
allow $c\in B(\mathfrak{H}\otimes\mathfrak{H})$, and define $\omega$ on the
latter algebra, even though our theory only needs it to be defined on the
algebraic tensor product $B(\mathfrak{H})\odot B(\mathfrak{H})$.

We begin by obtaining a positive trace-class operator $\kappa_{n}$
corresponding to the set $Y_{n}$ for every $n$. Each $\kappa_{n}$ will be one
of three types, namely a (maximally) entangled type, a mixed type, or a
product type, each of which we now discuss in turn for any $n$.

First, the \emph{entangled type} (corresponding to an entangled pure state):
We set
\[
\Omega_{n}=\sum_{q\in Y_{n}}\sqrt{\rho_{q}}e_{q}\otimes e_{q}
\]
and
\[
\kappa_{n}=\Omega_{n}\Join\Omega_{n}=\sum_{p\in Y_{n}}\sum_{q\in Y_{n}}
\sqrt{\rho_{p}\rho_{q}}(e_{p}\Join e_{q})\otimes(e_{p}\Join e_{q})
\]
for all $n$. It is straightforward to verify that
\begin{equation}
\Tr(\kappa_{n})=\sum_{q\in Y_{n}}\rho_{q} \label{kappanSp}
\end{equation}
and
\begin{equation}
\omega_{n}(a\otimes1)=\omega_{n}(1\otimes a)=\sum_{q\in Y_{n}}\rho
_{q}\left\langle e_{q},ae_{q}\right\rangle \label{omeganBep}
\end{equation}
for all $a\in B(\mathfrak{H})$.

Secondly, the \emph{mixed type} (corresponding to a mixture of pure states):
Setting
\[
\kappa_{n}=\sum_{q\in Y_{n}}\rho_{q}(e_{q}\otimes e_{q})\Join(e_{q}\otimes
e_{q})=\sum_{q\in Y_{n}}\rho_{q}(e_{q}\Join e_{q})\otimes(e_{q}\Join e_{q})
\]
we again obtain Eqs. (\ref{kappanSp}) and (\ref{omeganBep}).

Thirdly, the \emph{product type}: Setting
\[
\kappa_{n}=d_{n}\otimes d_{n}
\]
where
\[
d_{n}:=\left(  \sum_{p\in Y_{n}}\rho_{p}\right)  ^{-1/2}\sum_{q\in Y_{n}}
\rho_{q}(e_{q}\Join e_{q})
\]
we yet again obtain Eqs. (\ref{kappanSp}) and (\ref{omeganBep}).

For each type we take
\[
\kappa_{n}=0
\]
if $Y_{n}$ is empty (this allows for a partition of $\mathbb{N}_{+}$ into a
finite number of non-empty subsets).

For each $n$, let $\kappa_{n}$ be any of the three types above. Then
$\kappa_{n}$ is indeed trace-class and positive, so setting
\begin{equation}
\omega_{n}(c)=\Tr(\kappa_{n}c) \label{omegan}
\end{equation}
for all $c\in B(\mathfrak{H}\otimes\mathfrak{H})$, we obtain a well-defined
positive linear functional $\omega_{n}$ on $B(\mathfrak{H}\otimes
\mathfrak{H})$. Then
\[
\omega:=\sum_{n=1}^{\infty}\omega_{n}
\]
converges in the norm of $B(\mathfrak{H}\otimes\mathfrak{H})^{\ast}$, since
$\left\|  \omega_{n}\right\|  =\omega_{n}(1)=\Tr(\kappa_{n})$,
so $\sum_{n=1}^{\infty}\left\|  \omega_{n}\right\|  =1$. Correspondingly,
\begin{equation}
\kappa:=\sum_{n=1}^{\infty}\kappa_{n} \label{kappaReeks}
\end{equation}
converges in the trace-class norm $\left\|  \cdot\right\|  _{1}$, since
$\sum_{n=1}^{\infty}\left\|  \kappa_{n}\right\|  _{1}=\sum_{n=1}^{\infty
}\Tr(\kappa_{n})=1$. Then it indeed follows that
\[
\omega(c)=\sum_{n=1}^{\infty}\Tr(\kappa_{n}c)=\Tr
(\kappa c),
\]
since $|\sum_{n=1}^{m}\Tr(\kappa_{n}c)-\Tr(\kappa
c)|\leq\left\|  \sum_{n=1}^{m}\kappa_{n}-\kappa\right\|  _{1}\left\|
c\right\|  $.

Furthermore $\omega(1)=\sum_{n=1}^{\infty}\omega_{n}(1)=\sum_{n=1}^{\infty
}\rho_{n}=1$, and from Eq. (\ref{omeganBep}) it follows that the conditions in
Eq. (\ref{vbKop1}) hold. So $\omega$ is a coupling of $\zeta$ with itself as required.

For $Y_{1}=$ $\mathbb{N}_{+}$, i.e. $\kappa=\kappa_{1}$, we can get two
extremes, namely the diagonal coupling $\omega$ if $\kappa_{1}$ is of the
entangled type, and the product state $\omega=\zeta\otimes\zeta$ on
$B(\mathfrak{H}\otimes\mathfrak{H})$ when $\kappa_{1}$ is of the product type.
But the construction above gives many cases other than these two extremes.
Then balance with respect to $\omega$ is non-trivial, but does not necessarily
force two systems $\mathbf{A}$ and $\mathbf{B}$ on the same algebra $A$ to
have the same dynamics as in Proposition \ref{id-kor}.

\subsection{The dynamics}\label{VbDin}

We now construct dynamics in order to obtain examples of systems on the von
Neumann algebra $B(\mathfrak{H})$. Let $r_{j}\in\{3,4,5,...\}$ and $0<k_{j}<1$
for $j=1,2,3,...$, and write $k=(k_{1},k_{2},k_{3},...)$. In terms of the
$n\times n$ matrix
\[
O_{n}=\left[
\begin{array}
[c]{cccc}
0 & \cdots & 0 & 1\\
1 &  &  & 0\\
& \ddots &  & \vdots\\
&  & 1 & 0
\end{array}
\right]  ,
\]
with the blank spaces all being zero, we then define $R_{k}\in B(\mathfrak{H})$ 
by the infinite matrix
\[
R_{k}=\left[
\begin{array}
[c]{ccc}
k_{1}^{1/2}O_{r_{1}} &  & \\
& k_{2}^{1/2}O_{r_{2}} & \\
&  & \ddots
\end{array}
\right]
\]
in the basis $e_{1},e_{2},e_{3},...$, where again the blank spaces are zero.
In other words, $R_{k}e_{1}=k_{1}^{1/2}e_{2}$ etc. So $R_{k}$ consists of a
infinite direct sum of finite cycles, each cycle including its own factor
$k_{n}^{1/2}$. Replacing $k$ by $1-k:=(1-k_{1},1-k_{2},1-k_{3},...)$, we
similarly obtain $R_{1-k}$. In the same basis we consider a self-adjoint
operator $g\in B(\mathfrak{H})$ defined by the diagonal matrix
\[
g=\left[
\begin{array}
[c]{ccc}
g_{1} &  & \\
& g_{2} & \\
&  & \ddots
\end{array}
\right] ,
\]
with $g_{1},g_{2},g_{3},...$ a bounded sequence in $\mathbb{R}$. Note that
$R_{k}^{\ast}R_{k}+R_{1-k}R_{1-k}^{\ast}=1$. So we can define the generator
$\mathcal{K}$ of a uniformly continuous semigroup $\mathcal{S}=(\mathcal{S}
_{t})_{t\geq0}$ in $B(\mathfrak{H})$ by%
\[
\mathcal{K}(a)=R_{k}^{\ast}aR_{k}+R_{1-k}aR_{1-k}^{\ast}-a+i[g,a]
\]
for all $a\in B(\mathfrak{H})$. See for example \cite[Corollary 30.13]{Par};
the original papers on generators for uniformly continuous semigroups are
\cite{GKS} and \cite{L}.

In the same way and still using the same basis, for $l=(l_{1},l_{2}
,l_{3},...)$ with $0<l_{j}<1$ we define the generator $\mathcal{L}$ of a
second uniformly continuous semigroup $\mathcal{T}=(\mathcal{T}_{t})_{t\geq0}$
in $\mathfrak{H}$ by
\[
\mathcal{L}(b)=R_{l}^{\ast}bR_{l}+R_{1-l}bR_{1-l}^{\ast}-b+i[h,b]
\]
for all $b\in B(\mathfrak{H})$, where the diagonal matrix
\[
h=\left[
\begin{array}
[c]{ccc}
h_{1} &  & \\
& h_{2} & \\
&  & \ddots
\end{array}
\right]  ,
\]
with $h_{1},h_{2},h_{3},...$ a bounded sequence in $\mathbb{R}$, defines a
self-adjoint operator $h\in B(\mathfrak{H})$.

In the rest of Section \ref{Vb}, we assume the following:
\begin{align*}
\rho_{1}  &  =...=\rho_{r_{1}}\\
\rho_{r_{1}+1}  &  =...=\rho_{r_{1}+r_{2}}\\
\rho_{r_{1}+r_{2}+1}  &  =...=\rho_{r_{1}+r_{2}+r_{3}}\\
&  \vdots
\end{align*}
Then the state $\zeta$ is seen to be invariant under both $\mathcal{S}$ and
$\mathcal{T}$ by checking that $\zeta\circ\mathcal{K}=0$ and $\zeta
\circ\mathcal{L}=0$.

It is going to be simpler (but equivalent) to work directly in terms of
$B(\mathfrak{H})$, rather than its cyclic representation. Nevertheless, since
much of the theory of this paper is expressed in the cyclic representation, it
is worth expressing the various objects in this representation as well. In
particular we can then see how to obtain duals directly in terms of
$B(\mathfrak{H})$.

Our two systems $\mathbf{A}$ and $\mathbf{B}$, viewed in the cyclic
representation, are in terms of $A=B=\pi(B(\mathfrak{H}))$, with the dynamics
given by
\[
\alpha_{t}(\pi(a))=\pi(\mathcal{S}_{t}(a))
\]
and
\[
\beta_{t}(\pi(b))=\pi(\mathcal{T}_{t}(b))
\]
and the states $\mu$ and $\nu$ both given by
\[
\mu(\pi(a))=\nu(\pi(a))=\zeta(a)=\Tr(\rho a)
\]
for all $a,b\in B(\mathfrak{H})$. The diagonal coupling for $\mu$
\[
\delta_{\mu}:\pi(B(\mathfrak{H}))\odot\pi^{\prime}(B(\mathfrak{H}))\rightarrow\mathbb{C}
\]
is given by
\begin{align*}
\delta_{\mu}(\pi(a)\odot\pi^{\prime}(b))  &  =\left\langle \Omega,\pi
(a)\pi^{\prime}(b)\Omega\right\rangle =\left\langle \Omega,(a\otimes
b)\Omega\right\rangle \\
&  =\sum_{p=1}^{\infty}\sum_{q=1}^{\infty}\left\langle e_{p},\rho^{1/2}
ae_{q}\right\rangle \left\langle e_{q},\rho^{1/2}b^{T}e_{p}\right\rangle \\
&  =\operatorname{Tr}(\rho^{1/2}a\rho^{1/2}b^{T})
\end{align*}
where $b^{T}\in B(\mathfrak{H})$ is obtained as the transpose of the matrix
representation of $b$ in terms of the basis $e_{1},e_{2},e_{3},...$. In effect
$\delta_{\mu}$ is the maximally entangled state $\left\langle \Omega
,(\cdot)\Omega\right\rangle $ on $B(\mathfrak{H})\odot B(\mathfrak{H})$,
reducing to $\operatorname{Tr}(\rho(\cdot))$ on $B(\mathfrak{H})$.

The dual $\beta_{t}^{\prime}:$ $\pi^{\prime}(B(\mathfrak{H}))\rightarrow
\pi^{\prime}(B(\mathfrak{H}))$ of $\beta_{t}$ is given by
\[
\left\langle \Omega,\pi(b)\beta_{t}^{\prime}(\pi^{\prime}(b^{\prime}
))\Omega_{\zeta}\right\rangle =\left\langle \Omega,\beta_{t}(\pi
(b))\pi^{\prime}(b^{\prime})\Omega\right\rangle
\]
for all $b,b^{\prime}\in B(\mathfrak{H})$.

We therefore define the dual $\mathcal{L}^{\prime}$ of $\mathcal{L}$ via the
representations by requiring
\[
\left\langle \Omega,\pi(b)\pi^{\prime}(\mathcal{L}^{\prime}(b^{\prime}
))\Omega\right\rangle =\left\langle \Omega,\pi(\mathcal{L}(b))\pi^{\prime
}(b^{\prime})\Omega\right\rangle
\]
for all $b,b^{\prime}\in B(\mathfrak{H})$, i.e.
\[
\Tr(\rho^{1/2}a\rho^{1/2}(\mathcal{L}^{\prime}(b))^{T})=\Tr(\rho^{1/2}
\mathcal{L}(a)\rho^{1/2}b^{T})
\]
for all $a,b\in B(\mathfrak{H})$. Note that $\mathcal{L}^{\prime}$ is indeed 
the dual (with respect to $\zeta$) of $\mathcal{L}$ in the sense of Theorem 
\ref{duaalBestaan}, but represented on $\mathfrak{H}$ instead of on the GNS Hilbert 
space. It is then straightforward to verify that
\begin{equation}
\mathcal{L}^{\prime}(b)=R_{1-l}^{\ast}bR_{1-l}+R_{l}bR_{l}^{\ast}-b+i[h,b]
\label{L'}
\end{equation}
for all $b\in B(\mathfrak{H})$. From this one can see that $\mathcal{L}
^{\prime}$ is also the generator of a uniformly continuous semigroup
$\mathcal{T}^{\prime}=(\mathcal{T}_{t}^{\prime})_{t\geq0}$ in $\mathfrak{H}$,
which in addition satisfies
\[
\left\langle \Omega,\pi(b)\pi^{\prime}(\mathcal{T}_{t}^{\prime}(b^{\prime
}))\Omega\right\rangle =\left\langle \Omega,\pi(\mathcal{T}_{t}(b))\pi
^{\prime}(b^{\prime})\Omega\right\rangle
\]
and therefore
\[
\pi^{\prime}(\mathcal{T}_{t}^{\prime}(b^{\prime}))=\beta_{t}^{\prime}
(\pi^{\prime}(b^{\prime}))
\]
for all $b,b^{\prime}\in B(\mathfrak{H})$. As with $\mathcal{L}^{\prime}$ above, 
$\mathcal{T}_{t}^{\prime}$ is the dual of $\mathcal{T}_{t}$ in the sense of 
Definition \ref{duaalDef}, but represented on $\mathfrak{H}$. So we 
correspondingly call the semigroup $\mathcal{T}^{\prime}$ the \emph{dual} 
of the semigroup $\mathcal{T}$. 

We now have a complete description of the systems, as well as their duals.

\subsection{Balance}\label{VbBal}

We now show examples of balance between
\[
\mathbf{A:}=(B(\mathfrak{H}),\mathcal{S},\zeta)\text{ \ \ and \ \ }
\mathbf{B:}=(B(\mathfrak{H}),\mathcal{T},\zeta)
\]
and illustrate a number of points made in this paper. Remember that since we
now have a continuous time parameter $t\geq0$, the balance condition in
Definition \ref{balans} is required to hold at every $t$. However, it then
follows that $\mathbf{A}$ and $\mathbf{B}$ are in balance with respect to
$\omega$ if and only if
\[
\Tr(\kappa(\mathcal{K}(a)\otimes b))=
\Tr(\kappa(a\otimes\mathcal{L}^{\prime}(b))
\]
for all $a,b\in B(\mathfrak{H})$. From this one can easily check that
$\mathbf{A}$ and $\mathbf{B}$ are in balance with respect to $\omega$ if and
only if
\begin{align*}
&  (R_{k}\otimes1)\kappa(R_{k}\otimes1)^{\ast}+(R_{1-k}\otimes1)^{\ast}
\kappa(R_{1-k}\otimes1)-i[g\otimes1,\kappa]\\
&  =(1\otimes R_{1-l})\kappa(1\otimes R_{1-l})^{\ast}+(1\otimes R_{l})^{\ast
}\kappa(1\otimes R_{l})-i[1\otimes h,\kappa]
\end{align*}
holds. However, equating the real and imaginary parts respectively (keeping in
mind that $\kappa$ as given in Subsection \ref{VbKop} is a real infinite
matrix in the basis $e_{p}\otimes e_{q}$), we see that this is equivalent to
\begin{align}
&  (R_{k}\otimes1)\kappa(R_{k}\otimes1)^{\ast}+(R_{1-k}\otimes1)^{\ast}
\kappa(R_{1-k}\otimes1)\nonumber\\
&  =(1\otimes R_{1-l})\kappa(1\otimes R_{1-l})^{\ast}+(1\otimes R_{l})^{\ast
}\kappa(1\otimes R_{l}) \label{kappaBalans}
\end{align}
and
\begin{equation}
\lbrack g\otimes1,\kappa]=[1\otimes h,\kappa] \label{kappaBalansKom}
\end{equation}
both being true.

To proceed, we refine the construction of $\kappa$ in Subsection \ref{VbKop},
by only allowing
\[
Y_{n}=\bigcup_{p\in I_{n}}Z_{p}
\]
where $Z_{1}=\{1,2,...,r_{1}\}$, $Z_{2}=\{r_{1}+1,r_{1}+2,...,r_{1}+r_{2}\}$,
etc., and where $I_{1},I_{2},I_{3},...$ is any sequence of disjoint subsets of
$\mathbb{N}_{+}$ such that $\cup_{n\in\mathbb{N}_{+}}I_{n}=\mathbb{N}_{+}$.
Note that an $I_{n}$ is allowed to be empty (then $Y_{n}$ is empty), and it is
also allowed to be infinite.

It then follows that $\mathbf{A}$ and $\mathbf{B}$ are in balance with respect
to $\omega$ if and only if
\begin{align}
&  (R_{k}\otimes1)\kappa_{n}(R_{k}\otimes1)^{\ast}+(R_{1-k}\otimes1)^{\ast
}\kappa_{n}(R_{1-k}\otimes1)\nonumber\\
&  =(1\otimes R_{1-l})\kappa_{n}(1\otimes R_{1-l})^{\ast}+(1\otimes
R_{l})^{\ast}\kappa_{n}(1\otimes R_{l}) \label{kappanBalans}
\end{align}
and
\begin{equation}
\lbrack g\otimes1,\kappa_{n}]=[1\otimes h,\kappa_{n}] \label{kappanBalansKom}
\end{equation}
both hold for every $n$. To see that Eq. (\ref{kappanBalans}) and Eq.
(\ref{kappanBalansKom}) follow from Eq. (\ref{kappaBalans}) and Eq.
(\ref{kappaBalansKom}) respectively, place the latter into $\left\langle
e_{p}\otimes e_{q},(\cdot)e_{p^{\prime}}\otimes e_{q^{\prime}}\right\rangle $
for $p,q,p^{\prime},q^{\prime}\in Y_{n}$. The converse holds, since Eq.
(\ref{kappaReeks}) is convergent in the trace-class norm.

To evaluate these conditions in detail is somewhat tedious, so we just
describe it in outline below.

Note that, roughly speaking, in a term like 
$(R_{k}\otimes1)\kappa_{n}(R_{k}\otimes1)^{\ast}$, for $\kappa_{n}$ 
of the entangled or mixed type, the first slot in the tensor product structure 
of $\kappa_{n}$ is advanced by one step in each cycle appearing in $R_{k}$. 
In a term like $(1\otimes R_{l})^{\ast}\kappa_{n}(1\otimes R_{l})$, on the other 
hand, the second slot is rolled back by one step in each cycle, which is 
equivalent to the first slot being advanced by one step. So, if $\kappa_{n}$ 
is of the entangled or mixed type, and
\begin{equation}
k_{p}=l_{p} \label{k=l}
\end{equation}
for each $p\in I_{n}$, then Eq. (\ref{kappanBalans}) holds.

Conversely, for $p\in I_{n}$, note from the definitions of the entangled and
mixed type $\kappa_{n}$ that since $r_{p}>2$, the terms $(R_{k}\otimes
1)\kappa_{n}(R_{k}\otimes1)^{\ast}$ and $(1\otimes R_{l})^{\ast}\kappa
_{n}(1\otimes R_{l})$ have to be equal (hence $k_{p}=l_{p}$), for Eq.
(\ref{kappanBalans}) to hold; the terms $(R_{1-k}\otimes1)^{\ast}\kappa
_{n}(R_{1-k}\otimes1)$ and $(1\otimes R_{1-l})\kappa_{n}(1\otimes
R_{1-l})^{\ast}$ involve other basis elements of $\mathfrak{H}\otimes
\mathfrak{H}$ and therefore can not ensure Eq. (\ref{kappanBalans}) when
$(R_{k}\otimes1)\kappa_{n}(R_{k}\otimes1)^{\ast}\neq(1\otimes R_{l})^{\ast
}\kappa_{n}(1\otimes R_{l})$.

For the product type $\kappa_{n}$, Eq. (\ref{kappanBalans}) always holds,
since $\kappa_{n}$ then commutes with $R_{k}\otimes1$ and $1\otimes R_{l}$.

When $\kappa_{n}$ is of the entangled type, one can verify by direct
calculation that Eq. (\ref{kappanBalansKom}) holds if and only if
\begin{equation}
g_{p}-g_{q}=h_{p}-h_{q} \label{g=h}
\end{equation}
for all $p,q\in Y_{n}$. For the other two types of $\kappa_{n}$, Eq.
(\ref{kappanBalansKom}) always holds, since then $\kappa_{n}$, $g\otimes1$ and
$1\otimes h$ are diagonal, so the commutators are zero.

We conclude that $\mathbf{A}$ and $\mathbf{B}$ are in balance with respect to
$\omega$ if and only if the following is true: Eq. (\ref{k=l}) holds for all
$p\in I_{n}$ for every $n$ for which $\kappa_{n}$ is either of the entangled
or mixed type, and Eq. (\ref{g=h}) holds for all $p\in I_{n}$ for every $n$
for which $\kappa_{n}$ is of the entangled type.

We now also have an example where the transitivity in Theorem \ref{trans} is
trivial, meaning that $\omega\circ\psi=\mu\odot\xi^{\prime}$ despite having
$\omega\neq\mu\odot\nu^{\prime}$ and $\psi\neq\nu\odot\xi^{\prime}$. To see
this, let $\mathbf{C}$ be a system constructed in the same way as $\mathbf{A}$
and $\mathbf{B}$ above, so it has the same von Neumann algebra and state, but
the generator giving its dynamics can use different choices in place of $k,g$
and $l,h$. As above, construct two couplings $\omega$ and $\psi$ (giving
balance of $\mathbf{A}$ and $\mathbf{B}$ with respect to $\omega$, and of
$\mathbf{B}$ and $\mathbf{C}$ with respect to $\psi$), but with entangled and
mixed types not in overlapping parts of the two couplings respectively (i.e.
the respective $Y_{n}$ sets corresponding to these two types in the respective
couplings should be disjoint), while the rest of each coupling is a
$\kappa_{n}$ of the product type. Then it can be verified using Proposition
\ref{trivtranskar} that we indeed obtain $\omega\circ\psi=\mu\odot\xi^{\prime
}$, despite having $\omega\neq\mu\odot\nu^{\prime}$ and $\psi\neq\nu\odot
\xi^{\prime}$. This illustrates that to have $\omega\circ\psi\neq\mu\odot
\xi^{\prime}$, we need sufficient ``overlap'' between $\omega$ and $\psi$,
where this overlap condition has been made precise in Hilbert space terms (in
the cyclic representations) by Proposition \ref{trivtranskar}.

\subsection{A reversing operation}\label{VbOmk}

Here we consider $\Theta$-sqdb in Definition \ref{sqdb} and Corollary
\ref{fb&balans}, as well as the two generalized detailed balance conditions
suggested at the end of Section \ref{fb}. Take $\Theta$ to be transposition in
the basis $e_{1},e_{2},e_{3},...$, i.e.%
\[
\Theta(a):=a^{T}
\]
for all $a\in B(\mathfrak{H})$. This is the standard choice of a reversing
operation for $(B(\mathfrak{H}),\zeta)$, used for example in \cite[Section
2]{FR}. In the cyclic representation, $\Theta$ would be given by
$\pi(a)\mapsto\pi(a^{T})$. It is readily confirmed from Eq. (\ref{jB(h)}) that
in this case the $\Theta$-KMS dual of $\mathbf{B}$ is $\mathbf{B}^{\Theta
}=(B(\mathfrak{H}),\mathcal{T}^{\prime},\zeta)$, i.e. in the cyclic
representation we would have $\alpha_{t}^{\Theta}=\alpha_{t}^{\prime}$ for all
$t$.

For the diagonal coupling $\delta$, obtained when $\kappa_{1}$ is of the
entangled type with $Y_{1}=\mathbb{N}_{+}$, then from Eqs. (\ref{k=l}) and
(\ref{L'}) we see that $\mathbf{B}$ and $\mathbf{B}^{\Theta}$ are in balance
with respect to $\delta$, i.e. $\mathbf{B}$ satisfies $\Theta$-sqdb (Corollary
\ref{fb&balans}), if and only if $l_{p}=1-l_{p}$, i.e. $l_{p}=1/2$, for all
$p$.

More generally, consider the situation where $\mathbf{B}$ satisfies $\Theta
$-sqdb, and $\mathbf{A}$ and $\mathbf{B}$ are in balance with respect to
$\omega$. It then follows from Eq. (\ref{k=l}) that $k_{p}=1/2$ for all $p$ in
every $I_{n}$ such that $\kappa_{n}$ is of the entangled or mixed type, but we
need not have $k_{p}=1/2$ for other values of $p$. This is therefore a
strictly weaker condition on $\mathbf{A}$ than $\Theta$-sqdb, as long as not
all the $\kappa_{n}$ are of the entangled or mixed type.

Next consider the situation where $\mathbf{A}$ and $\mathbf{A}^{\Theta}$ are
in balance with respect to $\omega$, where again not all the $\kappa_{n}$ are
of the entangled or mixed type. Then in a similar way we again see that
$k_{p}=1/2$ for all $p$ in every $I_{n}$ such that $\kappa_{n}$ is of the
entangled or mixed type, but we need not have $k_{p}=1/2$ for other values of
$p$. So again this is a strictly weaker condition than $\Theta$-sqdb.

This illustrates the two conditions suggested at the end of Section \ref{fb},
albeit in a very simple situation. Here the two conditions are essentially
equivalent when applied to $\mathbf{A}$, but we expect this not to be the case
in general.

\section{Further work}\label{afdVw}

From Subsection \ref{VbOmk} we see, in a specific example, that a system 
$\mathbf{A}$ in balance with its $\Theta$-KMS dual $\mathbf{A}^{\Theta}$, 
where $\Theta$ is a reversing operation, can possibly be heuristically 
interpreted as satisfying $\Theta$-sqdb in some respects, since we had 
$k_{p}=1/2$ for some values of $p$, but not necessarily all. However, 
this special case does not give a physical interpretation of balance 
in general.

Theorem \ref{balkar} gives a hint toward a general interpretation, namely that 
if $\mathbf{A}$ and $\mathbf{B}$ are in balance with respect to $\omega$, then 
the dynamics of system $\mathbf{A}$ is partially carried over to system $\mathbf{B}$.
However, a physical interpretation of balance in general can possibly
be made more precise.

Now, as seen in particular from Theorem \ref{balkar}, balance seems to indicate 
some common structure in the two involved systems.
However, this is a subtle issue. Already in the classical case, in the context
of joinings, it has been shown that (translating into our context) two systems
can be nontrivially in balance (i.e. the coupling is not the product state),
while the two systems have no ``factor'' (roughly speaking a subsystem) in
common. This was a difficult problem in classical ergodic theory posed by
Furstenberg in \cite{F67} in 1967, and was only solved a decade later by
Rudolph in \cite{Rud}. Therefore we suspect that balance between two systems
is more general than the existence of some form of common system inside the
two systems. This issue has not been pursued in this paper, but appears worth 
investigating.

It also seems natural to study joinings directly for systems as defined in
Definition \ref{stelsel}. The idea would be to replace the balance conditions
in Definition \ref{balans}, by the joining conditions (possibly adapted
slightly) described in Remark \ref{bind}.

In principle we can view $E_{\omega}$ as a quantum channel. It could be of
interest to see what the physical significance of this map is, considering the
well-known correspondence between completely positive maps and bipartite
states in finite dimensions (see \cite{Ch}, but also \cite{dP} and \cite{J}
for earlier related work) which is of some importance in quantum information
theory. See for example \cite{VV}, \cite{AP} and \cite{JLF}. Some related work
has appeared in infinite dimensions for $B(H)$ and $B(H_{1},H_{2})$ as well
\cite{BQ2}, \cite{GKM}. Also see \cite[Section 1]{BCM} for further remarks.

Transitivity, via $E_{\psi}\circ E_{\omega}$, appears to be a basic ingredient
of the theory of balance, but we have not explored its consequences in this
paper. What are the physical implications or applications of transitivity?

In Section \ref{fb} we only considered standard quantum detailed balance with
respect to a reversing operation. It certainly seems relevant to investigate
if balance can be successfully used to give generalized forms of other conditions.

Furthermore, if balance can indeed be used to formulate certain types of
non-equilibrium steady states, as asked in Section \ref{fb}, then it seems
natural to connect this to entanglement and correlated states more generally.
Can results on entangled states be applied to a coupling $\omega$ of $\mu$ and
$\nu$ to study or classify certain classes of non-equilibrium steady states
$\mu$ (or $\nu$) of quantum systems? Note that the two extremes are the
product state $\omega=\mu\odot\nu^{\prime}$, which is the bipartite state with
no correlations, and the diagonal coupling $\delta_{\mu}$ of $\mu$ with
itself, which can be viewed as the bipartite state which is maximally
entangled while having $\mu$ and $\mu^{\prime}$ as its reduced states, at
least in the situation in Section \ref{Vb}.

We have only studied one example in this paper (in Section \ref{Vb}). To gain
a better understanding of balance, it is important to explore further
examples, especially physical examples, in particular in relation to non-equilibrium.

Lastly we mention the dynamical, weighted and generalized detailed balance
conditions studied in \cite{AI}, \cite{AFQ} and \cite{AFQ2} respectively,
along with a local KMS-condition, which was explored further in \cite{AFid}
and \cite{FidV}. We suspect that it would be of interest to explore if there
are any connections between these, and balance as studied in this paper.

\section*{Acknowledgments}

We thank V. Crismale, F. Fidaleo, J. M. Lindsay, W. A. Majewski and A. G.
Skalski for discussions and remarks related to Sections \ref{afdKar} and
\ref{fb} of the paper. We also thank the anonymous referees for suggestions to
improve certain points in the exposition. This work was supported by the National 
Research Foundation of South Africa.

\end{document}